\documentclass[11pt,a4paper,USenglish]{article}
\usepackage[utf8]{inputenc}
\usepackage{type1cm}

\usepackage[bookmarks=false, %
 linkcolor=red!60!black]{hyperref}

\usepackage{amsmath}
\usepackage{amsfonts}
\usepackage{amssymb}
\usepackage{amsthm}

\usepackage{fullpage}

\usepackage{paralist}

\usepackage{xspace}

\usepackage[ruled,vlined]{algorithm2e}
\usepackage{pgfplots}

\usetikzlibrary{patterns}
\tikzstyle{node}=[circle, inner sep = 0pt, minimum size = 0.7em,fill=white,draw]
\tikzstyle{arc}=[]
\usetikzlibrary{shapes.multipart}

\usepackage{fancybox}

\usepackage{authblk}

\usepackage[textsize=footnotesize]{todonotes}

\renewcommand{\vec}{\mathaccent"017E }

\theoremstyle{plain}
\newtheorem{theorem}{Theorem}
\newtheorem{lemma}{Lemma}

\theoremstyle{definition}
\newtheorem{example}{Example}

\newenvironment{lemenum}{\begin{compactenum}[(i)]}{\end{compactenum}}

\newcommand{\mfts}{\textsc{Minimum Feasible Tileset}\xspace}
\newcommand{\gmfts}{\textsc{Generalized Minimum Feasible Tileset}\xspace}

\newcommand{\pdsm}{\textsc{Exact Cover by $d$-Sets}\xspace}

\newcommand{\containment}{$\mathsf{NP} \subseteq \mathsf{coNP/poly}$\xspace}
\newcommand{\NP}{\ensuremath{\mathsf{NP}}\xspace}
\newcommand{\APX}{\ensuremath{\mathsf{APX}}\xspace}

\newcommand{\Oh}{\mathcal{O}}

\newcommand{\h}{\ensuremath{\mathcal{H}}\xspace}
\newcommand{\I}{\ensuremath{\mathcal{I}}\xspace}
\newcommand{\J}{\ensuremath{\mathcal{J}}\xspace}

\newcommand{\N}{\ensuremath{\mathbb{N}}\xspace}
\renewcommand{\P}{\ensuremath{\mathcal{P}}\xspace}

\newcommand{\R}{\ensuremath{\mathcal{R}}\xspace}
\renewcommand{\S}{\ensuremath{\mathcal{S}}\xspace}
\newcommand{\T}{\ensuremath{\mathcal{T}}\xspace}

\newcommand{\partition}{\Pi}

\newcommand{\introduceproblem}[3]{%
\begin{quote}
#1\\
\textbf{Input:} #2\\
\textbf{Problem:} #3
\end{quote}}

\newcommand{\cref}[1]{(\ref{#1})\xspace}

\title{The Minimum Feasible Tileset problem\thanks{An extended abstract of this article appeared at the \emph{12th Workshop on Approximation and Online Algorithms}, Wroc\l{}aw, September 2014~\cite{DKS14}. In comparison, apart from full proof details, the present article additionally contains examples and an APX-hardness proof. The authors gratefully acknowledge support by the Alexander von Humboldt-Foundation (Yann Disser), the `Excellence Initiative' of the German Federal and State Governments and the Graduate School CE at TU Darmstadt (Yann Disser), the German Research Foundation (DFG), projects KR 4286/1 (Stefan Kratsch) and NI 369/12 (Manuel Sorge), the Israel Science Foundation, grant no. 551145/14 (Manuel Sorge), and the People Programme (Marie Curie Actions) of the European Union's Seventh Framework Programme (FP7/2007-2013) under REA grant agreement number 631163.11 (Manuel Sorge).}}

\author[1]{Yann Disser}

\affil[1]{Institut für Mathematik, Graduate School CE, TU Darmstadt, Germany\\\texttt{disser@mathematik.tu-darmstadt.de}}

\author[2]{Stefan Kratsch}

\affil[2]{Institut für Informatik, Humboldt-Universität zu Berlin, Berlin, Germany\\\texttt{kratsch@informatik.hu-berlin.de}}

\author[3]{Manuel Sorge}

\affil[3]{Institut für Softwaretechnik und Theoretische Informatik, TU Berlin, Germany and Department of Industrial Engineering and Management, Ben-Gurion University of the Negev, Beer Sheva, Israel\\\texttt{sorge@post.bgu.ac.il}}

\begin{document}

\maketitle

\begin{abstract}
 We introduce and study the \mfts problem: Given a set of symbols and subsets of these symbols (scenarios), find a smallest possible number of pairs of symbols (tiles) such that each scenario can be formed by selecting at most one symbol from each tile.
 We show that this problem is \APX-hard and that it is \NP-hard even if each scenario contains at most three symbols. Our main result is a 4/3-approximation algorithm for the general case.
 In addition, we show that the \mfts problem is fixed-parameter tractable both when parameterized with the number of scenarios and with the number of symbols.
\end{abstract}
\section{Introduction}\label{section:introduction}

Consider the general assignment problem where several devices (e.g., workers, robots, microchips, \ldots) each can be used in one of~$k$ functions/modes at a time (e.g., employing different skills, tools, instruction sets, \ldots). Given a set of scenarios, the goal is to assign~$k$ different functions to each device, such that, for each scenario, all functions requested by the scenario are available simultaneously. In this paper, we initiate the study of this problem for~$k=2$ and the case that each function is requested at most once by each scenario. Formally, we study the following problem (we use ``\emph{tile}'' instead of ``\emph{device}'' to intuitively capture the fact that a device/tile has two modes/sides).

\introduceproblem{\mfts}
{A universe of symbols $F$, scenarios $\S \subseteq 2^F \setminus \{F\}$.}%
{Find a minimum-size tileset~\T that is feasible for all scenarios in~$\S$.}
In the above, a \emph{tile} is a two-element subset of~$F$ and we refer to (multi-)sets of tiles as \emph{tilesets}. 
A tileset~\T is \emph{feasible} for scenario $S$ if we can produce all symbols in $S$ by taking at most one symbol from each tile in \T. Formally, a tileset~\T is feasible for a scenario~$S \subset F$ if there is a mapping~$\phi \colon \T \to F$, such that~$\phi(T) \in T$ for all~$T \in \T$, and~$S\subseteq \phi[\T] := \{\phi(T)\mid T\in\T\}$. 
By definition, no scenario contains all symbols of $F$. Note that such a scenario would require $|F|$ tiles, making the problem trivial.
Similarly, we may assume that all symbols in~$F$ appear in at least one scenario, otherwise we can simply remove each symbol that does not occur in any scenario.
Finally, the requirement that tiles contain no less than two symbols can be met by arbitrarily assigning a second symbol to all tiles of cardinality one. 

\newcommand{\tile}[2]{%
  \begin{tikzpicture}%
    \node[rectangle split, rectangle split parts=2, draw, anchor=center] at (0, 0) {\ensuremath{#1}\nodepart{two}\ensuremath{#2}};
  \end{tikzpicture}%
}
\begin{example}\label{ex:intro}
  Let us illustrate the problem with two instances of \mfts: 

  If our set of symbols is $F = \{A, B, C, 1, 2, 3\}$ and our set~$\S$ of
  scenarios consists of $\{A, B, 1, 2\}$, $\{A, C, 1, 3\}$, and
  $\{B, C, 2, 3\}$, then it is not hard to check that a feasible
  tileset is
  \begin{center}
    \tile{A}{B}, \tile{B}{C}, \tile{1}{2}, \tile{2}{3}.
  \end{center}
  Herein, each tile is represented by two adjoined boxes which
  correspond to the two modes in which we can use the tile. Clearly,
  the feasible tileset above is also of minimum size since each
  scenario requires at least four tiles.
  
  If we have $n \in \mathbb{N}$, $F = \{1, \ldots, 3n\}$, and our scenarios are all size-2
  subsets of~$F$, that is, $\S = \binom{F}{2}$, then a feasible tileset with $2n$ tiles is
  \begin{center}
    \tile{1}{2}, \tile{2}{3}, \hspace{.125cm} \tile{4}{5}, \tile{5}{6}, \hspace{.125cm} \ldots, \hspace{.125cm} \tile{3i - 2}{3i - 1}, \tile{3i -1}{3i}, \hspace{.125cm} \ldots, \hspace{.125cm} \tile{3n - 2}{3n - 1}, \tile{3n - 1}{3n}.
  \end{center}
  To see that the above tileset is feasible for $\binom{F}{2}$, pick
  any integer~$a \in \{1, \ldots, 3n\}$. If $a \neq 3i - 1$ for each
  $i \in \{1, \ldots, n\}$, then, whichever tile we pick to
  produce~$a$, each other integer from $1$ to $3n$ is available on
  some other tile. If $a = 3i - 1$ for some $i \in \{1, \ldots, n\}$,
  then, to produce~$a$ and any other
  integer~$b \in \{1, \ldots, 3n\} \setminus \{a\}$, we can pick an arbitrary tile
  for~$b$ and $a$ will be available on another tile. Hence, we can
  produce any scenario containing two distinct integers from 1 to $3n$
  by two distinct tiles. It will become clear later, that each
  feasible tileset for the instance~$(F, \binom{F}{2})$ of \mfts\
  contains at least~$2n$ tiles.
\end{example}

Apart from practical motivations, \mfts is appealing from a structural point of view. 
In this work we exhibit equivalent definitions for the problem which are interesting in their own right. 
At first glance, \mfts is a covering problem since we must cover all scenarios using tiles that can each cover one of the tile's two symbols in each scenario. It turns out that the problem can also be phrased as a packing/partitioning problem, but with an objective function different from the classical one in terms of number of packed objects or sets (see Section~\ref{section:hardness}). In addition, having tiles be symbol sets of size two suggests a graph interpretation where we are asked to find a minimum set of edges such that for 
each scenario
there is an orientation where each vertex has indegree at least one. 
For our presentation however, we favor the tileset formulation, since it
most naturally generalizes to the original assignment problem with tiles of larger sizes and scenarios which contain multiple copies of the same symbols.
Also, the \mfts interpretation appears suitable for studying the effect of parameters, such as the number of symbols/scenarios, on the complexity.

\paragraph{Results and Outline.}
\looseness=-1 We analyze the structure of the graph that has the tiles of a minimum cardinality tileset as its edges, and show that this graph is always (wlog.) a forest. In fact, only the component structure of this forest matters: We may replace trees by arbitrary trees spanning the same components without affecting the feasibility of the corresponding tileset~(\autoref{section:structure}). This lets us view \mfts as a partitioning problem, which in turn allows us to prove \NP-completeness even when scenarios have size at most three and lets us show \APX-hardness for the general case~(\autoref{section:hardness}). As our main result, we complement the hardness with a 4/3-approximation algorithm (for scenarios of arbitrary sizes) inspired by the component structure of the optimum solution~(\autoref{section:approximation}). We show that the problem is fixed-parameter tractable with respect to 
the number of scenarios~(\autoref{section:numberofscenarios}) and the number of symbols~(\autoref{section:numberofsymbols}), respectively. We also observe that, when each scenario has size at most~$d$, a polynomial-time compression of an arbitrary instance to $O(|F|^d)$ bits is possible without loosing the information about the size of the optimum solution and such a compression to $O(|F|^{d - \epsilon})$ bits is unlikely (\autoref{section:numberofsymbols}). Finally, we provide a preliminary result on the relevant variant of \mfts\ where the scenarios are multisets rather than sets and show that also this case is fixed-parameter tractable with respect to the number of symbols (\autoref{section:numberofsymbols}).

\paragraph{Related Work.}
The problem most closely related to \mfts is arguably \textsc{Set Packing}, as~$3$-\textsc{Set Packing} appears as a subproblem in our approximation algorithm and also as the source problem for our \NP-hardness reduction (in the form of \textsc{Maximum Three-Dimensional Matching}). \textsc{Set Packing} has been extensively studied for both approximability and parameterized complexity  (see, e.g.,~\cite{BansalCS09,Cygan13,SviridenkoW13} and~\cite{DellM12,Koutis08} for some recent results). The main difference between the two problems is that \textsc{Set Packing} is a maximization problem whereas \mfts seeks to minimize the size of a feasible tileset --- a measure that is only indirectly related to the number of sets~(scenarios). In particular, \textsc{Set Packing} becomes trivial for a bounded number of sets, whereas for \mfts we get a nontrivial polynomial-time algorithm via integer linear programming (see \autoref{section:numberofscenarios}).

As alluded to above, the \mfts problem can equivalently be seen as designing an edge-minimal graph on the set of symbols such that, for each scenario, the edges (tiles) can be oriented in such a way that all symbols in the scenario have indegree at least one.
The question whether a \emph{given} graph admits an orientation with certain properties has been studied in various settings. For example, Biedl et al.~\cite{Biedl/05} proposed an approximation algorithm for finding a balanced acyclic orientation. Another natural constraint on an orientation that has been studied is to prescribe degrees for each vertex~\cite{DisserMatuschke/16,Frank/76,Hakimi/65}.

More abstractly, we are looking for a graph on the set of symbols that fulfills a certain constraint for each scenario. The case where the subgraph induced by each scenario has to be connected is well-studied~\cite{BKMSV11,CKNSSW15,DuM88,GG13,JP87,BKKNS16}. In particular, it is \NP-hard to find the minimum number of edges needed~\cite{DuM88} and to decide whether a planar solution~\cite{BKMSV11,JP87} or a solution of treewidth at most three~\cite{GG13} exists.

\paragraph{Preliminaries.} For some positive
integer~$\ell \in \mathbb{N}$, we denote
$[\ell] := \{1, 2, \ldots, \ell\}$. For a set family~$\S$ we use $\bigcup \S$ as a shorthand for $\bigcup_{S \in \S} S$. Apart from standard Landau notation for running times, we also use the $O^*$ notation, which disregards factors that are polynomial in the input size. We use standard graph
notation, see the book by Diestel~\cite{reinhard_diestel_graph_2016}, for
example. For the relevant notions from parameterized complexity and
approximation complexity we refer to textbooks
\cite{flum_parameterized_2006,marek_cygan_parameterized_2015,rodney_g._downey_fundamentals_2013}
and refs.~\cite{Cre97,WS00}, respectively.

\section{Graph structure of tilesets}\label{section:structure}

The tiles in a tileset $\T$ over a universe of symbols $F$ can be viewed as the edges of the undirected (multi\nobreakdash-)~graph $G(\T):=(F,\T)$. In this section, we establish that there always exist optimal tilesets with a simple graph structure. This is made formal in the following lemma which will be useful in later sections.

\begin{lemma}\label{lem:struc}
  Let~$F$ be a universe of symbols, $\S$ a family of scenarios over~$F$, and $\T$~a tileset feasible for~$\S$. There is a tileset~$\T' \subseteq \binom{F}{2}$ feasible for~$\S$ such that $|\T'| \leq |\T|$ and~$G(\T')$ is a forest.
\end{lemma}
Note that each connected component of~$G(\T')$ has size at least two because
each symbol occurs in at least one scenario and hence is incident with at least one edge.

In the proof of \autoref{lem:struc} it is convenient to think of feasibility of~$\T$ via orientations of the graph~$G(\T)$. Let us say that an orientation of~$G(\T)$ is \emph{feasible} for the scenario~$S$ if each vertex in~$S$ has indegree at least one. It is easy to see that deciding whether \T is feasible for some scenario $S \subset F$ is equivalent to deciding whether there is a \emph{feasible} orientation of the edges of $G(\T)$ for $S$. We obtain the following lemma.

\begin{lemma}\label{lem:feasible-orientation}
  For every tileset $\T$ and scenario $S$ over a universe of symbols $F$ the following are equivalent.
  \begin{lemenum}
  \item $\T$ is feasible for~$S$,\label{en:fo1}
  \item there is a feasible orientation of~$G(\T)$ for~$S$,\label{en:fo2}
  \item for every connected component~$C$ of $G(\T)$, there is a feasible orientation of~$G(\T)$ for~$S \cap C$,\label{en:fo3}
  \item for every connected component~$C$ of $G(\T)$, tileset~$\T$ is feasible for~$S \cap C$.\label{en:fo4}
  \end{lemenum}
\end{lemma}

\begin{proof}
  Note that it suffices to prove the equivalence of the first three statements, since equivalence of \cref{en:fo1} and \cref{en:fo2} implies equivalence of \cref{en:fo3} and \cref{en:fo4}.
  
  $\text{\cref{en:fo1}} \Rightarrow \text{\cref{en:fo2}}$: Assume that~$\T$ is feasible for~$S$ and let $\phi:\T \to F$ be the corresponding mapping with $\phi(T) \in T$ for all $T \in \T$, and $S\subseteq \phi[\T]$. We obtain an orientation of~$G(\T)$ by orienting each edge~$T\in\T$ towards~$\phi(T)$. Since $S\subseteq \phi[\T]$, each symbol in $S$ has indegree at least one and we have a feasible orientation of $G(\T)$ for $S$.

  $\text{\cref{en:fo2}} \Rightarrow \text{\cref{en:fo3}}$: Clearly, a feasible orientation of~$G(\T)$ for $S$ is, in particular, a feasible orientation for~$S \cap C$ for every connected component~$C$ of~$G(\T)$.

  $\text{\cref{en:fo3}} \Rightarrow \text{\cref{en:fo1}}$: Let $C_1, \ldots, C_k$ denote the connected components of $G(\T)$ and assume that there are feasible orientations $\vec{G}_1,\dots,\vec{G}_k$ for~$S \cap C_1,\dots, S \cap C_k$, respectively. Since $S \subseteq \bigcup_i C_i$, we obtain a feasible orientation $\vec{G}$ of $G(\T)$ for $S$ as $\vec{G}_1[C_1] \uplus \cdots \uplus \vec{G}_k[C_k]$. We define the mapping $\phi:\T \to F$ by setting $\phi(T)=s$ for each $T \in \T$, where $s$ is the symbol towards which the edge $T$ is oriented in $\vec{G}$. By definition, $\phi(T)\in T$ for all $T \in \T$, and, since $\vec{G}$ is feasible for $S$, we have $S \subseteq \phi[\T]$. The existence of the mapping~$\phi$ hence proves that \T is feasible for $S$.
\end{proof}

Using the notion of feasible orientations we now observe that connected components in~$G(\T)$ yield feasibility for each of their strict subsets.

\begin{lemma}\label{lem:tree-feasible}
  Let $\T$ be a tileset, $C$ a connected component of~$G(\T)$ and~$C' \subsetneq C$. Then, $\T$ is feasible for~$C'$.
\end{lemma}
\begin{proof}
  The proof is by induction over the size of~$C'$. If~$C'$ contains a single symbol, that is, $C' = \{s\}$, then we obtain a feasible orientation by orienting an arbitrary edge towards~$s$; such an edge exists because~$C'$ is part of the (larger) connected component~$C$. 
  Consider the case~$|C'| > 1$. First assume that $G(\T)[C']$ contains no edges, i.e., $C'$ is an independent set. Then, there is an edge in $G(\T)[C]$ for each symbol $s \in C'$, connecting $s$ to $C \setminus C'$. A feasible orientation for~$C'$ can simply be obtained by orienting these edges towards~$C'$. 
  Now, assume there is an edge~$\{s, s'\}$ in~$G(\T)[C']$ and consider the graph~$G'$ obtained by contracting~$\{s, s'\}$. By induction, there is a feasible orientation of~$G'$ for~$C''$, where $C''$ is obtained from~$C'$ by identifying~$s$ and~$s'$. Hence, there is an orientation of~$G(\T)$ such that all vertices in~$C'$ except one of~$\{s, s'\}$ have indegree at least one. We orient the edge~$\{s, s'\}$ towards the vertex of smaller indegree to obtain the desired feasible orientation of~$G(\T)$.
\end{proof}

We are ready for a proof of \autoref{lem:struc}. Intuitively, we observe that cycle components in~$G(\T)$ yield feasibility for \emph{any} of their subsets and hence are a safe replacement for every component with a large number of edges. Then we show how to break cycle components into trees.

\begin{proof}[Proof of \autoref{lem:struc}]
  We replace connected components in~$G(\T)$, maintaining feasibility of~$\T$ and without increasing the cardinality of~$\T$.
  
  Let~$C_1, \dots, C_k$ be  the connected components of~$G(\T)$ that contain a cycle, and let~$C = \bigcup_{i = 1}^k C_i$. We obtain a new tileset $\T'$ by replacing the edges in~$G(\T)[C]$ with a cycle on~$C$. (If $|C| = 1$, we introduce a self-loop and if $|C| = 2$, we introduce two parallel edges.) Since each component $C_i$, $i \in [k]$,  originally contained at least $|C_i|$ edges, and we only introduced $|C|$ new edges, the cardinality of~$\T'$ is not larger than the cardinality of~$\T$. 
We observe that~$\T'$ is still feasible. Consider an arbitrary scenario~$S \in \S$. Clearly, if~$S \cap C = \emptyset$, then $\T'$~is feasible for~$S$. Hence, assume~$S \cap C \neq \emptyset$. By \autoref{lem:feasible-orientation} there is a feasible orientation of~$G(\T)$ for~$S \setminus C$. This implies that there is still a feasible orientation of~$G(\T')$ for~$S \setminus C$. By \autoref{lem:feasible-orientation} it suffices to prove that there is a feasible orientation of~$G(\T')$ for~$S \cap C$. Indeed, since~$G(\T')[C]$ is a cycle, orienting the edges in one direction along the cycle yields a feasible orientation for any subset of $C$, and, in particular, for~$S \cap C$. Hence, $\T'$ is still feasible for every~$S \in \S$.

  By definition, every connected component~$C'$ of~$G(\T)$ outside of~$C$ is a tree. Hence, the connected components of $G(\T')$ are trees, with the exception of at most one component~$C$ that is a cycle. We now modify~$C$ in order to obtain our final feasible tileset~$\T''$ with the desired structure. 

  First, consider the case that~$C$ is the only connected component of~$G(\T')$. Then, we can remove an arbitrary tile from~$\T'$ to obtain~$\T''$: Since, by definition, $\S$ does not contain $F$ as a scenario, by \autoref{lem:tree-feasible}, $\T''$~is feasible for all scenarios~$S \in \S$. Clearly, $G(\T'')$ is a tree, as required. 
  
  Now assume that there is at least one tree component~$C'$ in~$G(\T')$ along with~$C$. Consider an arbitrary edge~$\{s, s'\}$ in~$C$ and an arbitrary vertex~$s'' \in C'$. We remove~$\{s, s'\}$ from~$\T'$ and instead add the edge~$\{s, s''\}$ to obtain the tileset~$\T''$. Clearly, $G(\T'')$ is a forest. It remains to prove that~$\T''$ is feasible for every scenario~$S \in \S$. By \autoref{lem:feasible-orientation}, $\T'$~is feasible for~$S \setminus C$ and, hence, so is~$\T''$. Because $C \cup C'$ is a connected component in~$G(\T'')$, \autoref{lem:tree-feasible} guarantees that $\T''$ is feasible for every~$S' \subsetneq (C \cup C')$ and, in particular, $\T''$ is feasible for~$S \cap C$. Hence, as~$\T''$ is feasible for~$S \cap C$ and~$S \setminus C$, applying \autoref{lem:feasible-orientation} we obtain that~$\T''$ is feasible for~$S$. Clearly, $G[\T'']$ is a forest, as required.
\end{proof}

Intuitively, Lemmas~\ref{lem:feasible-orientation} and~\ref{lem:tree-feasible} imply that only the partition of the symbols induced by the component structure of the graph of a tileset matters, but not the exact topology of each of the trees. This leads to the following.

\begin{theorem}\label{cor:only_components_matter}
  Let \S be a family of scenarios and \T be a tileset over symbols~$F$. If $G(\T)$ is a forest, then \T~is feasible for~\S if and only if no connected component~$C$ of~$G(\T)$ is fully contained in any scenario~$S\in\S$, i.e., $C \nsubseteq S$ for all scenarios~$S\in\S$ and all connected components~$C$ of~$G(\T)$.
\end{theorem}

\begin{proof}
  ``$\Rightarrow$'': Assume towards a contradiction that~\T is feasible, $G(\T)$ is a forest, and there is a scenario~$S \in \S$ and a component~$C$ of $G(\T)$ such that~$C \subseteq S$. By \autoref{lem:feasible-orientation} there is a feasible orientation of $G(\T)$ for~$S \cap C = C$. But this is absurd, because~$G(\T)$ is a forest and hence $G(\T)[C]$ contains only $|C| - 1$ edges.
  
  ``$\Leftarrow$'': For each component~$C$ of~$G(\T)$ and every scenario~$S\in\S$, we have that $C \cap S \subsetneq C$, since $C \nsubseteq S$. By \autoref{lem:tree-feasible}, we get that~\T is feasible for~$C \cap S$. Since this is true for all choices of~$C$ and~$S$, \autoref{lem:feasible-orientation} implies that~\T is feasible for~\S.
\end{proof}

\begin{example}
  We can now observe that for the instance
  $(F = \{1, \ldots, 3n\}, \S = \binom{F}{2})$ from \autoref{ex:intro}
  each feasible tileset contains at least~$2n$ tiles: Let $\T$ be a
  feasible tileset for $\S$. No connected component of $G(\T)$ can be
  contained in any set in~$\S$. Thus, each connected component has
  size at least three, meaning that there are at least~$n$ connected
  components. Since each such connected component induces at least two
  tiles, $\T$ contains at least $2n$ tiles.
\end{example}

\section{NP-hardness and APX-Hardness of Minimum Feasible Tileset}\label{section:hardness}

In this section we establish the following result.

\begin{theorem}\label{thm:restricted-hard}
  \mfts{} is \APX-hard. \mfts\ is \NP-hard even if each scenario has size at most three.
\end{theorem}
\looseness=-1 Before proving \autoref{thm:restricted-hard}, let us check that the decision variant of \mfts, in which we want to check for feasible tilesets of size at most a given integer, is contained in \NP: A feasible tileset can be encoded using polynomially many bits with respect to~$|F|$. Verifying feasibility comes down to solving one bipartite matching problem for each scenario on an auxiliary graph that has an edge between each symbol in the scenario and every tile containing that symbol, which is possible in polynomial time. Thus we can infer from \autoref{thm:restricted-hard} that the decision variant of \mfts\ is \NP-complete.

We now prove \NP- and \APX-hardness of \mfts. For this, we first give a relation of \mfts\ to a partitioning problem. Let us say that, for a finite set of symbols~$F$ and a family of scenarios $\S \subseteq 2^F$, a partition~$\mathcal{P}$ of $F$ is \emph{admissible}, if for every $P \in \mathcal{P}$ and every $S \in \S$ we have $P \not\subseteq S$. We obtain the following.

\begin{lemma}\label{lem:partition-equiv}
  Let $F$ be a set of symbols and $\S \subseteq 2^F\setminus\{F\}$ a family of scenarios. There is a feasible tileset of size $\ell$ for $\S$ if and only if there is a partition of~$F$ which is admissible for $\mathcal{S}$ and comprises $|F| - \ell$~parts.
\end{lemma}

\newcommand{\fcp}{\textsc{Fine Constrained Partition}\xspace}

\begin{proof}

  ``$\Rightarrow$'': By \autoref{lem:struc} there is a feasible tileset~$\T'$ for~$\S$ of cardinality~$\ell$ such that~$G(\T')$ is a forest. The connected components~$C_1, \ldots, C_k$ of~$G(\T')$ induce a partition~$\mathcal{P}'$ which we claim to be admissible: Indeed, by \autoref{cor:only_components_matter} we have $C_i \nsubseteq S$ for all connected components~$C_i$, $i \in [k]$, and scenarios~$S \in \S$. Furthermore, since there are exactly $\ell$~edges in~$G(\T')$ and each connected component is a tree, we have $\ell = \sum_{i = 1}^k|C_i| - 1 = |F| - k$. 
  Hence, our partition has $k = |F| - \ell$ parts, as required.

  ``$\Leftarrow$'': Let~$\mathcal{P} = \{P_1, \ldots, P_p\}$ be an admissible partition with $|F| - \ell$~parts. We construct a tileset~$\T$ by setting~$G(\T)[P_i]$ to an arbitrary spanning tree for each~$i \in [p]$. Since $P_i \nsubseteq S$ for each~$S \in \S$ and each~$i \in [p]$, by \autoref{cor:only_components_matter}, $\T$ is feasible for~$\S$. 
  The number of tiles in~$\T$ is $\sum_{i = 1}^{p} |P_i| - 1 = |F| - p = |F| - (|F| - \ell) = \ell$, as required.
\end{proof}

\noindent Thus, \mfts\ is equivalent to finding a finest-possible partition, i.e. with maximum number of parts, of
the symbols such that no part in the partition is contained in any
scenario.

\newcommand{\bmtdm}{\textsc{Maximum Bounded 3\nobreakdash-Dimensional Matching}\xspace}
We now give a reduction from \bmtdm which is both \NP-hard and \APX-hard~\cite{Kann91}:

\introduceproblem{\bmtdm}
{Three pairwise disjoint sets~$X, Y, Z$, and a set $D \subseteq X \times Y \times Z$ of triples such that each element in $X \cup Y \cup Z$ occurs in at most three triples in~$D$.}
{Find a maximum-size \emph{three-dimensional matching} for $D$, i.e., a maximum-size subset~$D' \subseteq D$ such that no element of $X \cup Y \cup Z$ occurs in two triples in $D'$.}
\newcommand{\opttdm}{\ensuremath{\operatorname{opt}_{\text{3DM}}}}
\newcommand{\optft}{\ensuremath{\operatorname{opt}_{\text{FT}}}}
\begin{proof}[Proof of \autoref{thm:restricted-hard}]
  We give a PTAS-reduction from \bmtdm~\cite{Cre97}.  More precisely,
  given an instance $(X, Y, Z, D)$ and a desired approximation
  ratio~$r$, we construct an instance $(F, \S)$ of \mfts in time
  polynomial in the instance size for every fixed~$r$, subject to the
  following condition. 
  There is a function~$f \colon (0, 1) \to \mathbb{Q}$ such that for every $r \in (0,1)$ we have  $f(r) > 1$ and for a arbitrary given $f(r)$-approximate feasible tileset~$\T$ for $(F, \S)$ we can construct an $r$-approximate three-dimensional matching~$M$ for $(X, Y, Z, D)$ in
  polynomial time. We specify the function $f$ below (while ensuring that
  $f(r) > 1$).

  We first describe how to construct the \mfts instance $(F, \S)$ from
  a \bmtdm instance $(X, Y, Z, D)$. Set the universe
  $F := X \cup Y \cup Z$. We choose a function
  $g \colon (0, 1) \to \mathbb{N}$ with $g(r) \geq 3$ for all $r \in (0,1)$. 
  The precise function is given below. The
  scenarios $\S$ consist of all subsets of~$F$ that have size at most~$g(r)$ and do not
  contain any triple in $D$ as a subset (we interpret~$D$ as a family of
  three-element sets). Formally, $\S := \{S \subseteq F \mid |S| \leq g(r) \wedge \forall E \in D \colon E \setminus S \neq \emptyset\}$. This concludes the construction. Let
  $n := | X \cup Y \cup Z| = |F|$. Clearly, for any fixed $r$, we can carry out the construction
  in time $\Oh(n^{g(r) + 1})$, that is, in polynomial time in the
  instance size.

  Before we show how to compute an approximate three-dimensional matching from an
  approximate feasible tileset, we find a relation between the optimal
  solution sizes of the two instances. Note that, from each
  three-dimensional matching~$N$ we can construct an admissible
  partition~$\R$ of $F$ for $\S$ satisfying $|\R| = |N|$ as
  follows. Initially, take $\R = N$ (where $N$ is interpreted as a
  family of three-element sets). Then, replace an arbitrary part
  $P \in \R$ with $P \cup ((X \cup Y \cup Z) \setminus \bigcup
  \R)$. That is, add to $P$ all elements not covered by~$N$. By
  definition of $\S$, there is no set $S \in \S$ that contains any
  $P \in \R$, and hence $\R$ is admissible. Letting \opttdm\ denote the
  size of an optimal solution to the \bmtdm\ instance, we thus have
  $|\P^*| \geq \opttdm$ for an admissible partition~$\P^*$ containing
  the maximum number of parts. Lemma~\ref{lem:partition-equiv} implies
  that $|\P^*| = n - \optft$, where $\optft$ denotes the size of an optimal solution for the
  \mfts\ instance. Rearranging terms hence yields $\optft = n -
  |\P^*|$. Because of $|\P^*| \geq \opttdm$, we have
  \begin{equation}
    \optft\ \leq\ n - \opttdm\label{eq:optsols}.
  \end{equation}

  Now let $\T$ be an arbitrary $f(r)$-approximate feasible
  tileset~$\T$ for $(F, \S)$. We construct an $r$-approx\-i\-mate
  three-dimensional matching~$M$ for $(X, Y, Z, D)$ in polynomial time
  as follows. Along the way, we gather observations that allow us to
  prove that $M$ is $r$-approximate in the end.

  First, by
  Lemma~\ref{lem:partition-equiv} there is a partition~$\P$ of $F$
  which is admissible for $\S$ and has $n - |\T|$ parts. In other
  words, $\T$ has $n - |\P|$ tiles. (As the proof of Lemma~\ref{lem:partition-equiv} is constructive, it is not hard to check that $\P$
  can be computed in polynomial time.) As $\T$ is $f(r)$-approximate,
  $|\T| \leq f(r)\optft$ and hence, $n - |\P| \leq f(r)\optft$. Applying
  Inequality~\eqref{eq:optsols} we thus obtain
  \begin{equation}
    \label{eq:ptasadmiss}
    n - |\P| \ \leq\  f(r)(n - \opttdm).
  \end{equation}
  We create a partition~$\P^1 = \P^1_3 \cup \P^1_{g(r) + 1}$ from~$\P$ as follows. Obtain
  $\P^1_3$ by picking, for each part~$P \in \P$ which contains a triple of $D$, one triple $E \in D$ with $E \subseteq P$ and putting $E$ (as a set) into~$\P^1_3$. To create $\P^1_{g(r) + 1}$, if $|F \setminus \bigcup_{P \in P^1_3}P| < g(r) + 1$, then put $F \setminus \bigcup_{P \in P^1_3}P$ into $\P^1_{g(r) + 1}$ as a single set of size at most~$g(r)$. In this case we call $\P^1_{g(r) + 1}$ \emph{degenerate}. Otherwise, put $\P^1_{g(r) + 1}$ to be an
  arbitrary partition of $F \setminus \bigcup_{P \in \P^1_3}P$ into
  parts of size $g(r) + 1$ and, perhaps, one part of size at least~$g(r) + 2$ and at most $2g(r) + 1$. Note that $\P^1_3$ is admissible for~$\S$ because
  each triple in $\P^1_3$ is not contained in any set in $\S$ by
  definition of $\S$%
  . We claim that $|\P^1| \geq
  |\P|$. Clearly, for each part in $\P$ that contains a triple of $D$
  there is at least one part also in $\P^1$. Furthermore, since $\S$
  contains all $g(r)$-element sets which do not contain any triple
  of~$D$, each set $P \in \P$ that does not contain a triple from~$D$
  must contain at least~$g(r) + 1$ elements because $\P$ is
  admissible for~$\S$. Hence, $|\P^1| \geq |\P|$. From
  Inequality~\eqref{eq:ptasadmiss} it follows that
  \begin{equation}
    n - |\P^1| \ \leq \ f(r)(n - \opttdm).\label{eq:ptasmod}
  \end{equation}  
  Note that the three-element sets in $\P^1$, i.e., $\P^1_3$, form
  a three-dimensional matching for the instance $(X, Y, Z, D)$.
  The sets in $\P^1_3$
  will form our $r$-approximate three-dimensional matching~$M$ after
  one further augmentation step. The aim of this augmentation is to
  bound $|P^1_3|$ by a function of $|P^1_{g(r) + 1}|$%
  . This enables us to give a lower bound on $|M|$
  via the size of $\P^1$.

  Consider the following modification of~$\P^1$. If there is a triple
  in $D$ that is disjoint from $\bigcup_{P \in \P^1_3}P$, then add
  this triple to $\P^1_3$. If we now have $|F \setminus \bigcup_{P \in P^1_3}P| < g(r) + 1$, then replace the sets in $\P^1_{g(r) + 1}$ by the single set $F \setminus \bigcup_{P \in P^1_3}P$. Otherwise, replace $\P^1_{g(r) +1}$ by an arbitrary
  partition of $F \setminus \bigcup_{P \in \P^1_3}P$ into parts of
  size $g(r) + 1$ and, perhaps, one part of size at least~$g(r) + 2$ and at most~$2g(r) + 1$.

  We claim that the above two modification steps do not decrease the
  number of sets in~$\P^1$. This is clear if $\P^1_{g(r) + 1}$ was
  degenerate before applying them. Otherwise, we have
  $|F \setminus \bigcup_{P \in P^1_3}P| > g(r) + 3$ before applying
  the modification. Hence, each set in $\P^1_{g(r) + 1}$ had size at
  least $g(r) + 1 \geq 4$. Since we moved only three elements from
  these sets to $\P^1_3$ and afterwards repartition the remaining
  elements with sets of size at least 4, the total number of sets
  cannot decrease.

  As before, $\P_3^1$ remains admissible for~$\S$. Let $\P^2$ be the
  partition obtained by exhaustively applying the above modification,
  let $\P^2_3$ equal the resulting set~$\P^1_3$ and let
  $P^2_{g(r) + 1}$ be the resulting set~$\P^1_{g(r) + 1}$%
  . We define the three-dimensional matching $M$ as
  the family of three-element parts in~$\P^2 \cap D$, that is $M =
  \P^2_3$. As mentioned, we have $|\P^2| \geq |\P^1|$. From
  Inequality~\eqref{eq:ptasmod} it thus follows that
  \begin{equation}
    n - |\P^2| \ \leq \ f(r)(n - \opttdm).\label{eq:ptasmod2}
  \end{equation}
  It remains to show that $M$ is $r$-approximate (for appropriate
  functions $f(r)$ and $g(r)$).

  We now claim that $12 |P^2_3|/ (g(r) + 1) \geq |\P^2_{g(r) +
    1}|$. If $\P^1_{g(r) + 1}$ is degenerate and if this relation does
  not hold, then $|F|$ is upper bounded by
  $3 |\P^2_3| + g(r) - 1 < 3 (g(r) + 1) / 12 + g(r) - 1$, that is, $|F|$ is upper bounded by a
  constant. Hence, we may compute the optimal solution in constant
  time in this case. Thus, we may, without loss of generality, assume
  that the relation holds if $\P^1_{g(r) + 1}$ is degenerate.

  If $\P^1_{g(r) + 1}$ is not degenerate, we claim that the above
  modification of $\P^1$ is applicable as long as
  $12 |\P^1_3| < (g(r) + 1) |\P^1_{g(r) + 1}|$, where the right hand
  side bounds the number of unmatched elements. Indeed, since each
  element in $F = X \cup Y \cup Z$ is contained in at most three
  triples in~$D$, for each triple~$E$ in $\P^1_3$, there are at most
  twelve elements of $F$ whose incident triples in $D$ cannot be added
  to $\P^1_3$, because they overlap with~$E$. Hence, if
  $12 |\P^1_3| < (g(r) + 1) |\P^1_{g(r) + 1}|$, then there exists at
  least one element of $F$ whose incident triples do not overlap with
  any triple in~$\P^1_3$. This means that at least one triple will be
  added to $\P^1_3$ in the above modification step, because, without
  loss of generality, each element is in at least one triple. This
  indeed implies for partition~$\P^2$ (after exhaustive modification) that
  $12 |P^2_3|/ (g(r) + 1) \geq |\P^2_{g(r) + 1}|$. We thus have
  \begin{align*}
    n - |\P^2| &\ =\  n - (|\P^2_3| + |\P^2_{g(r) + 1}|)\\
    & \ \geq\ n - \left(1 + \frac{12}{g(r) + 1}\right)|\P^2_3|,
  \end{align*}
  and since $\P^2_3 = M$, in combination
  with Inequality~\eqref{eq:ptasmod2} we have
  \begin{equation*}
    n - \left(1 + \frac{12}{g(r) + 1}\right)|M| \ \leq\ f(r)(n - \opttdm).
  \end{equation*}
  Thus,
  \begin{equation}
    \label{eq:ptasmbound}
    \frac{g(r) + 1}{g(r) + 13} \cdot((1 - f(r)) n + f(r) \opttdm) \ \leq\ |M|.
  \end{equation}
  The same modification that we applied above to (the three-element
  sets of) $\P^1$ works for each three-dimensional matching. As we
  cannot improve a maximum three-dimensional matching, each element of
  $F$ is either matched --- of this type there are $3\opttdm$
  elements --- or it is in a triple together with a matched element --- of
  this type there are at most $12\opttdm$ elements because each element is in at
  most three triples in~$D$. Hence,
  $n \leq 3 \opttdm + 12 \opttdm = 15 \opttdm$. Note that
  $1 - f(r) < 0$. Inequality~\eqref{eq:ptasmbound} thus implies
  \begin{equation}
    \label{eq:approxfactor}
    \frac{g(r) + 1}{g(r) + 13} \cdot (15 - 14f(r)) \cdot \opttdm \ \leq\ |M|.
  \end{equation}
  We now define $f$ and $g$ by setting $g(r) := \max\{3, \lceil 13r/(1 - r) \rceil\}$ and
  \begin{equation*}
    f(r) \ :=\ \frac{-rg(r) - 13r + 15g(r) + 15}{14g(r) + 14}.
  \end{equation*}
  Clearly, $g(r) \geq 3$ as required. We claim also that $f(r) > 1$. To see this, consider subtracting the denominator from the numerator in $f(r)$ to obtain~$x$, that is, 
  \begin{equation*}
    x\ :=\ -rg(r) - 13r + 15g(r) + 15 - 14g(r) - 14 \ =\ -rg(r) - 13r + g(r) + 1\ =\ (1 - r)g(r) - 13r + 1.
  \end{equation*}
  If $3 \leq 13r/(1 - r)$, then
  $x \geq (1 -r)\cdot\frac{13r}{1 - r} - 13r + 1 > 0$, that is
  $f(r) > 1$. If $3 > 13r/(1 - r)$ then $r < 3/16$. Thus,
  $x \geq (1 - 3/16) \cdot 3 - 13 \cdot 3/16 + 1 = 1 > 0$ and again $f(r) > 1$. Thus,
  these are suitable definitions. All that remains is to show that the approximation factor in Inequality~\eqref{eq:approxfactor} is at least~$r$, that is,
  \begin{equation}
    \label{eq:ptasfactor}
    \frac{g(r) + 1}{g(r) + 13} \cdot (15 - 14f(r)) \ \geq\ r.
  \end{equation}
  Observe that
  \begin{equation*}
    15 - 14f(r) \ =\ \frac{15g(r) + 15}{g(r) + 1} - \frac{15g(r) + 15 - rg(r) - 13r}{g(r) + 1} \ =\ \frac{rg(r) + 13r}{g(r) + 1}. %
  \end{equation*}
  Hence,
  \begin{equation*}
    \frac{g(r) + 1}{g(r) + 13} \cdot (15 - 14f(r))\ =\ \frac{rg(r) + 13r}{g(r) + 13} \ =\ r.
  \end{equation*}
This implies that Inequality~\eqref{eq:ptasfactor} holds. Hence, there is a PTAS-reduction
  from \bmtdm\ to \mfts.

  \NP-hardness of the decision version of \mfts\ follows from the
  following modification to the reduction above. Instead of \bmtdm\ we
  reduce from the \NP-hard decision problem which asks whether there
  is a three-dimensional matching with $n/3$
  triples~\cite{GareyJohnson79}. We use the reduction above and set
  $g(r) = 3$ and the desired feasible tileset size to~$2n/3$. As mentioned, this can be done in polynomial time. For the correctness, each
  three-dimensional matching of size~$n/3$ is also an admissible
  partition for $\S$. Hence, it implies a feasible tileset of size
  $n - n/3 = 2n/3$ by Lemma~\ref{lem:partition-equiv}. In the reverse direction, each feasible tileset of
  size $2n/3$ implies an admissible partition with $n/3$ parts by
  Lemma~\ref{lem:partition-equiv}. Each of these
  parts is of size three because $\S$ contains all size-two subsets of~$F$. Among sets of size three, the only sets not contained in $\S$ are precisely the sets in~$D$; hence, each part of an admissible partition is in~$D$. That is, any admissible partition is a
  three-dimensional matching as well.
\end{proof}

\section{A 4/3-approximation for Minimum Feasible Tileset}\label{section:approximation}

In this section, we propose an approximation algorithm for \mfts{} with unbounded scenario size. 
Motivated by the structural insights of \autoref{section:structure}, we construct a tileset that induces a forest in the corresponding graph, with the property that none of its components are contained in a single scenario.
Since a component of size~$k$ requires~$k-1$ tiles, we additionally aim for small components in order to keep the resulting tileset small.

We first take as many components of size two as possible among all disjoint sets of two symbols that are not both contained in the same scenario.
This can easily be achieved by computing a maximum matching in the graph that has an edge for each candidate component.
Similarly, among all remaining symbols, we try to form many (disjoint) components of size three, without creating components that are contained in a single scenario.
For this, we employ a simple greedy strategy, that repeatedly takes any possible component until no possible candidates remain.
(While there are better packing strategies available for sets of size three, we will see that improving the packing strategy alone does not improve our approximation ratio.)
Finally, for each leftover symbol we add an individual tile (pairing that symbol in such a way as to prevent cycles). 

We give a more formal listing in~\autoref{alg:approximation}. 
We use $\bar{F}_{i}(F')=\{C\in{F' \choose i}\mid\forall S\in\mathcal{S} \colon C\nsubseteq S\}$ to denote the family of all sets of symbols in $F'$ that are of size~$i$ and not fully contained in a single scenario.
In the following, we identify connected components with their sets of vertices.

\begin{algorithm}      
  \DontPrintSemicolon

\KwIn{A set~$F$ of symbols and a set~$\S$ of scenarios, where $\mathcal{S}\subseteq2^{F}\setminus \{F\}$.}

\KwOut{A set of tiles $\mathcal{T}$.}

$\mathcal{T}_{2}\leftarrow$ maximum matching in graph $G(\bar{F}_{2}(F))$.\;

\smallskip{}
$\mathcal{P}\leftarrow$ greedy set packing of $\bar{F}_{3}(F\setminus\bigcup_{t\in\mathcal{T}_{2}}t)$.\;

$\mathcal{T}_{3}\leftarrow\bigcup_{\{f_{1},f_{2},f_{3}\}\in\mathcal{P}}\{\{f_1,f_2\},\{f_2,f_3\}\}$.\;

\smallskip{}
\lIf{$\T_2\cup\T_3 \neq \emptyset$}
{ take $f_\mathrm{root} \in \bigcup_{t\in\T_2\cup\T_3}t$ }
\lElse
{ take $f_\mathrm{root} \in F$.\;}

$\T_1\leftarrow \{\{f,f_\mathrm{root}\}\mid f\in F \setminus \bigcup_{t\in\mathcal{T}_{2}\cup\mathcal{T}_{3}}t\ , f \neq f_\mathrm{root}\}$.\;

\smallskip{}
\Return{$\mathcal{T}=\mathcal{T}_{1}\cup\mathcal{T}_{2}\cup\mathcal{T}_{3}$.}\;

  \SetAlgoRefName{$A$}
  \caption{4/3-approximation for minimum feasible tilesets}
  \label{alg:approximation}
\end{algorithm}
\begin{theorem}
  \autoref{alg:approximation} computes a 4/3-approximation for \mfts{}.
\end{theorem}

\begin{proof}
\looseness=-1 We first argue that the set of tiles $\mathcal{T}=\mathcal{T}_{1}\cup\mathcal{T}_{2}\cup\mathcal{T}_{3}$
computed by \autoref{alg:approximation} is feasible for \S.
First observe that~$G(\T)$ is a forest.
This is true, because $G(\T_2\cup\T_3)$ consists of trees of sizes~2 and~3, $G(\T_1)$ is a star, and $\T_1\cap(\T_2\cup\T_3)$ contains at most one node~$(f_\mathrm{root})$.
Using \autoref{cor:only_components_matter} it only remains to show that no connected component~$C$ of $G(\T)$ is contained in any scenario~$S \in \S$, i.e. $C \cap S \subsetneq C$. By definition of \autoref{alg:approximation} this is true for all connected components of the graph $G(\T_2\cup\T_3)$. If~$\T_2\cup\T_3 \neq \emptyset$, then each component of~$G(\T)$ is a superset of a component of~$G(\T_2\cup\T_3)$, and is thus not contained in any scenario.
If~$\T_2\cup\T_3$ is empty, then~$G(\T) = G(\T_1)$ consists of a single component that is not contained in any scenario, since, by definition,~$F \notin \S$. Thus \T is feasible for~\S.

We now bound the size of \T with respect to a minimum cardinality tileset~$\mathcal{T}^{\star}$. To do this we \emph{distribute} virtual currency\emph{ (gold)} to the symbols in $F$, such that the total gold distributed is $4/3$ times the size of $\T^{\star}$. 
We later use this gold to \emph{pay} one unit of gold to certain symbols that these can in turn use to \emph{provide for} (at most) one tile of $\mathcal{T}$ that involves this symbol. 
To complete the proof, we establish that each tile of $\mathcal{T}$ is provided for by one of its two symbols.

Let $G^{\star}:=G(\mathcal{T}^{\star})$ be the graph induced by
$\mathcal{T}^{\star}$ and $\bar{F}_{i}^{\star}$ be the set of connected
components of size $i\in\{2, \ldots, |F|\}$ in $G^{\star}$. By \autoref{lem:struc}, we may assume that $G^{\star}$ is a forest. Furthermore, because each symbol appears in at least one scenario, graph~$G^{\star}$ does not contain components of size~1.
Since the symbols in a component of size $i>1$ are part of exactly $i-1$
tiles in $\mathcal{T}^{\star}$, we may distribute all available gold by giving
$4/3\cdot\frac{i-1}{i}$ gold to each symbol in a component of $\bar{F}_{i}^{\star}$,
for all $i\in\{2,\dots,|F|\}$.
This gold is used to pay symbols in what follows. We call a symbol $s \in F$ \emph{sufficiently paid} if one of the following holds: (i)~$s$~is paid, (ii)~$s$~appears in a tile $T\in\mathcal{T}_{2}$
and the other symbol of $T$ is paid, or (iii)~$s$~appears in a tile~$T\in\mathcal{T}_{3}$ and the other two symbols in the same component
of $G(\mathcal{T}_{3})$ are paid. Below, we show how to sufficiently pay all symbols. This completes the proof, since then all tiles in $\T_1 \cup \T_2 \cup \T_3$ can be provided for (note that then each tile in $\T_1$ contains its own paid symbol). We call a component of~$G^{\star}$ \emph{sufficiently paid}, if all its symbols are sufficiently paid. Let~$F_{\geq 4}^{\star} := F \setminus \bigcup_{C\in\bar{F}_{2}^{\star}\cup\bar{F}_{3}^{\star}}C$ be the set of all symbols not in components of size two or three in~$G^{\star}$. In paying the symbols we will maintain the invariant that each element of $\bar{F}_{2}^{\star}\cup\bar{F}_{3}^{\star} \cup F_{\geq 4}^{\star}$ is either sufficiently paid, or it still holds its gold (all its symbols still hold their gold, respectively).

We define a graph $H=(V,E)$ that has the components in $\bar{F}_{2}^{\star}\cup\bar{F}_{3}^{\star}$
as its vertices, as well as the symbols that are not part of these components,
i.e., $V=\bar{F}_{2}^{\star}\cup\bar{F}_{3}^{\star}\cup F_{\geq4}^{\star}$ (cf.~Figure~\ref{fig:graphG}).
In this way, each vertex of $H$ represents up to three symbols.
For each tile $T\in\mathcal{T}_{2}$ we introduce an edge connecting
the vertices of $H$ representing the two symbols of~$T$, possibly introducing self-loops.
Since $\mathcal{T}_{2}$ is a matching, and since the vertices in~$H$ represent at most three symbols each, all vertices in~$H$ have
degree at most~3.
We partition the edges of~$H$ into paths, cycles, and self-loops, and show for each how to use the gold remaining at its vertices to pay all symbols in the components of~$G^{\star}$ that are intersected by the path/cycle/self-loop.
We will ensure that every symbol (except possibly~$f_\mathrm{root}$) on a tile in~$\T_1$ is paid.
Since each symbol on a tile of $\mathcal{T}_{2}$ appears
only exactly on this and no other tile of $\T_2\cup\T_3$, it
is thus sufficient to pay only one of the two symbols on each tile
of~$\mathcal{T}_{2}$.
\begin{figure}
\begin{centering}
\resizebox{0.7\textwidth}{!}{%
\begin{tikzpicture}
 
\def \dist {.5}
 
\newcommand{\tset}[6]{
 \begin{scope}[rotate around={#3:(#1,#2)}]
  \node (#4) at (#1,#2 - \dist) [node] {};    
  \node (#5) at (#1,#2) [node] {};      
  \node (#6) at (#1,#2 + \dist) [node] {};   
  \draw (#1,#2) ellipse (.3 and .8*2*\dist);
 \end{scope}
}

\newcommand{\dset}[5]{
 \begin{scope}[rotate around={#3:(#1,#2)}]
  \node (#4) at (#1,#2 - \dist/2) [node] {};    
  \node (#5) at (#1,#2 + \dist/2) [node] {};      
  \draw (#1,#2) ellipse (.3 and \dist);
 \end{scope}
}

\def \offX {1.5*\dist}; \def \offY {1.5*\dist};

\tset{0}{0}{0}{v0}{v1}{v2}
\tset{2*\offX}{-2*\offY}{90}{v3}{v4}{v5}
\tset{4*\offX}{0}{0}{v6}{v7}{v8}
\dset{2*\offX}{2*\offY}{90}{v9}{v10}
\dset{5.5*\offX}{0}{90}{v11}{v12}
\tset{8*\offX}{0}{90}{v13}{v14}{v15}
\tset{10*\offX}{0}{0}{v16}{v17}{v18}
\dset{11.5*\offX}{\dist + 1.5*\offY}{-45}{v19}{v20}
\tset{11.5*\offX}{-\dist - 1.5*\offY}{45}{v21}{v22}{v23}
\node (v24) at (13*\offX,-\dist - 3*\offY) [node] {};    

\draw [arc] (v0) to (v5);
\draw [arc] (v3) to (v6);
\draw [arc] (v8) to (v9);
\draw [arc] (v10) to (v2);
\draw [arc] (v7) to (v12);
\draw [arc] (v11) to (v15);
\draw [arc] (v13) to (v17);
\draw [arc] (v18) to (v19);
\draw [arc] (v16) to (v23);
\draw [arc] (v21) to (v24);

\dset{9*\offX}{6*\offY}{65}{v25}{v26}
\tset{11*\offX}{5*\offY}{65}{v27}{v28}{v29}
\draw [arc] (v25) to (v29);
\dset{13*\offX}{4*\offY}{65}{v30}{v31}
\draw [arc] (v31) to (v27);
\dset{14.7*\offX}{3.15*\offY}{65}{v32}{v33}
\draw [arc] (v30) to (v33);

\tset{5.5*\offX}{2.5*\offY}{-65}{v34}{v35}{v36}
\tset{8*\offX}{3.75*\offY}{-65}{v37}{v38}{v39}
\draw [arc] (v36) to (v37);

\tset{3*\offX}{4*\offY}{75}{v}{v}{v}

\dset{7*\offX}{-3*\offY}{55}{v40}{v41}
\draw [arc, bend right=120, looseness=4] (v40) to (v41);

\tset{13.5*\offX}{-.5*\offY}{-15}{v}{v}{v}
\node (v) at (0.5*\offX,5*\offY) [node] {};

\node (v) at (5.5*\offX,5.5*\offY) [node] {};
\begin{scope}[shift={(5.5*\offX,5.5*\offY)}]
\foreach [count=\n] \x in {72,144,...,360} {
        \node (v) at (\x:\dist) [node, anchor=-90+\x] {};
};
\end{scope}

\node (v) at (8*\offX,2*\offY) [node] {};

\begin{scope}[rotate around={45:(9*\offX,-2*\offY)}]
\node (v) at (9*\offX - \dist/2,-2*\offY + \dist/2) [node] {};
\node (v) at (9*\offX - \dist/2,-2*\offY - \dist/2) [node] {};
\node (v) at (9*\offX + \dist/2,-2*\offY + \dist/2) [node] {};
\node (v) at (9*\offX + \dist/2,-2*\offY - \dist/2) [node] {};
\end{scope}

\node (v) at (3.5*\offX,-3*\offY) [node] {};

\end{tikzpicture}
}
\par\end{centering}
\caption{Illustration of the graph $H$ that has as its nodes the components of sizes~2 and~3 in $G^\star$ and all symbols that appear in components of other sizes. An edge between symbols corresponds to a tile in $\mathcal{T}_2$.
\label{fig:graphG}}
\end{figure}
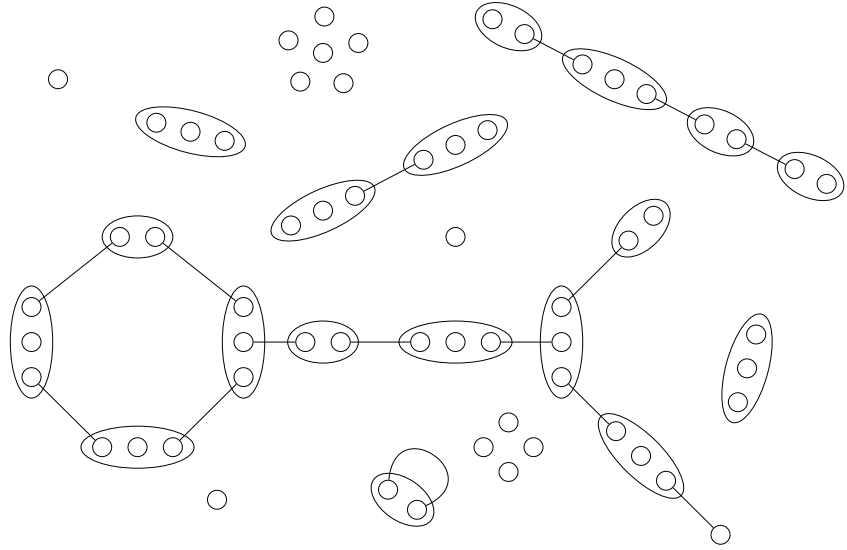

Let $\mathcal{P}$ be the set of all paths in $H$ connecting
(different) vertices of degree~1 or~3 with internal nodes of degree 2.
Consider the paths in $\mathcal{P}$ one by one. We use the
gold available along path~$P\in\mathcal{P}$ of length $k$ as follows (cf.~Figure~\ref{fig:paying_paths}).
Let $N_{2},N_{3}$ be the number of internal nodes of $P$ that represent
2 and 3 symbols, respectively. Note that $P$ has no inner nodes that
represent a single symbol, since $\mathcal{T}_{2}$ is a matching, and hence
$k=1+N_{2}+N_{3}$. Also, $P$ is the only path visiting these inner
nodes and hence they all still hold their gold. Let $N_{1}^{\mathrm{end}},N_{2}^{\mathrm{end}},N_{3}^{\mathrm{end}}\leq2$
be the number of endpoints of $P$ that still hold gold and represent
1, 2, and 3 symbols, respectively. Similarly, let $N_{0}^{\mathrm{end}}$
be the number of endpoints without gold. By our invariant, the symbols or components represented by the endpoints without gold left have already been sufficiently
paid before. We make sure that all other nodes along $P$ are
sufficiently paid. We do this by, for all tiles that form the
path~$P$, paying one of the two corresponding symbols, and, in addition,
paying \emph{every} further symbol represented by nodes along $P$. Note that this preserves the invariant. The
total cost is 
\begin{equation}
C^{-}=k+N_{2}^{\mathrm{end}}+2N_{3}^{\mathrm{end}}+N_{3}-N_{0}^{\mathrm{end}}=1+N_{2}^{\mathrm{end}}+2N_{3}^{\mathrm{end}}+N_{2}+2N_{3}-N_{0}^{\mathrm{end}}.
\label{eq:cminus}
\end{equation}
\begin{figure}
\begin{centering}
\resizebox{0.9\textwidth}{!}{%
\begin{tikzpicture}

\def \dist {.5}
 
\newcommand{\tset}[6]{
 \begin{scope}[rotate around={#3:(#1,#2)}]
  \node (#4) at (#1,#2 - \dist) [node] {};    
  \node (#5) at (#1,#2) [node] {};      
  \node (#6) at (#1,#2 + \dist) [node] {};   
  \draw (#1,#2) ellipse (.3 and .8*2*\dist);
 \end{scope}
}

\newcommand{\dset}[5]{
 \begin{scope}[rotate around={#3:(#1,#2)}]
  \node (#4) at (#1,#2 - \dist/2) [node] {};    
  \node (#5) at (#1,#2 + \dist/2) [node] {};      
  \draw (#1,#2) ellipse (.3 and \dist);
 \end{scope}
}

\newcommand{\tsetshaded}[6]{
 \begin{scope}[rotate around={#3:(#1,#2)}]
  \draw [pattern=north west lines] (#1,#2) ellipse (.3 and .8*2*\dist);
  \node (#4) at (#1,#2 - \dist) [node] {};    
  \node (#5) at (#1,#2) [node] {};      
  \node (#6) at (#1,#2 + \dist) [node] {};   
 \end{scope}
}

\def \offTD {3.8*\dist};
\def \offDD {3*\dist};
\def \offSD {2.2*\dist};

\tsetshaded{0}{0}{-90}{v}{v}{v1}
\dset{\offTD}{0}{-90}{v2}{v3}
\draw [arc] (v1) to (v2);
\dset{\offTD+\offDD}{0}{-90}{v4}{v5}
\draw [arc] (v3) to (v4);
\tset{2*\offTD+\offDD}{0}{-90}{v6}{v}{v7}
\draw [arc] (v5) to (v6);
\dset{3*\offTD+\offDD}{0}{-90}{v8}{v}
\draw [arc] (v7) to (v8);

\node (v9) at (0,1) [node] {};      
\dset{\offSD}{1}{-90}{v10}{v11}
\draw [arc] (v9) to (v10);
\dset{\offSD+\offDD}{1}{-90}{v12}{v13}
\draw [arc] (v11) to (v12);
\dset{\offSD+2*\offDD}{1}{-90}{v14}{v15}
\draw [arc] (v13) to (v14);
\tset{\offSD+2*\offDD+\offTD}{1}{-90}{v16}{v}{v}
\draw [arc] (v15) to (v16);

\end{tikzpicture}
}
\par\end{centering}
\caption{Illustration of our procedure for paying all symbols represented by the nodes along a path in graph~$H$. Shaded components have been sufficiently paid for previously. For the top path we have $N_1^{\mathrm{end}}=N_3^{\mathrm{end}}=1$, $N_2=3$, $N_0^{\mathrm{end}}=N_2^{\mathrm{end}}=N_3=0$. For the bottom path we have $N_0^{\mathrm{end}}=N_2^{\mathrm{end}}=1$, $N_2=2$, $N_3=1$, $N_1^{\mathrm{end}}=N_3^{\mathrm{end}}=0$.
\label{fig:paying_paths}}
\end{figure}
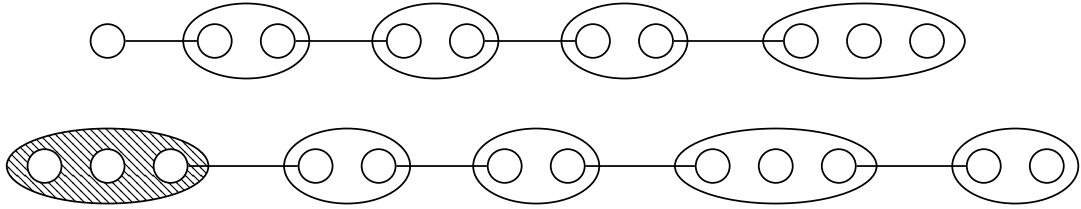

Using that each endpoint of $P$ that contributes to $N_1^{\mathrm{end}}$ represents a symbol that is part of a component in $G^{\star}$ of size~$i \geq 4$, we get that the gold available at this symbol is at least $\frac{4}{3}\cdot\frac{i - 1}{i} \geq 1$.
Hence, the gold available to us is at least
\begin{equation}
C^{+}=\frac{4}{3}(N_{2}^{\mathrm{end}}+2N_{3}^{\mathrm{end}}+N_{2}+2N_{3}+\frac{3}{4}N_{1}^{\mathrm{end}}).
\label{eq:cplus}
\end{equation}

Since $N_{0}^{\mathrm{end}}+N_{1}^{\mathrm{end}}+N_{2}^{\mathrm{end}}+N_{3}^{\mathrm{end}}=2$,
we get
\begin{equation*}
  C^+ - C^- = 1 - \frac{2}{3}N_2^{\mathrm{end}} - \frac{1}{3}N_3^{\mathrm{end}} + \frac{1}{3}N_2 + \frac{2}{3}N_3.
  \label{eq:cdiff}
\end{equation*}
Hence, we have $C^{+}\geq C^{-}$, unless $N_2^{\mathrm{end}} = 2$ and $N_{0}^{\mathrm{end}}=N_{1}^{\mathrm{end}}=N_{3}^{\mathrm{end}}=N_{2}=N_{3}=0$, i.e. $P$ is of length one, connecting two tiles~$p_1, p_2 \in \bar{F}_{2}^{\star}$ by an edge which corresponds to a tile~$t \in \T_{2}$. To see that this case cannot occur, observe that, first, $p_1$ and~$p_2$ are of degree 1 in~$H$. Second, since $\mathcal{T}^{\star}$
is feasible, no component of $G^{\star}$ is contained
in a single scenario (\autoref{cor:only_components_matter}), and thus~$p_1, p_2 \in \bar{F}_{2}^{\star}\subseteq\bar{F}_{2}(F)$. This is a contradiction to~$\T_{2}$ being a maximum
matching in graph~$G(\bar{F}_{2}(F))$, as the matching can be augmented by removing~$t$ and adding~$p_1$ and~$p_2$.

Similarly to the above, we can consider all cycles in $H$ with at
most one node of degree~3 one by one. (Note that cycles with at least two nodes of degree~3 contain a path as before.) 
If a cycle of length~$k$ does not contain a node of degree~3, or the node of degree~3 is not yet sufficiently paid (and thus still holds its gold), the cost for the cycle and its available
gold are 
\[
C^{-}=k+N_{3}=N_{2}+2N_{3}=\frac{3}{4}C^{+}<C^{+},
\]
where $N_{2},N_{3}$ are the numbers of nodes of $P$ that represent
2 and 3 symbols, respectively. If the node of degree 3 has no gold left, then
it has already been sufficiently paid and $C^{-} = N_{2}+2N_{3}-3 < \frac{3}{4}C^{+}$. 
In either case, the available gold allows
to sufficiently pay all nodes along the cycle. 
Finally, each self-loop in $H$ connects two symbols in the same component~$C$ of size~2 or 3 in $G^\star$.
If~$|C|=2$, the gold available among the two symbols is $C^+ = \frac{4}{3}$, while we require only $C^- = 1$ unit of gold.
If~$|C|=3$, we have $C^+ = \frac{8}{3}$ and $C^- = 2$. 

After processing all paths, cycles, and self-loops 
all nodes of $H$ intersecting a tile of $\mathcal{T}_{2}$ are sufficiently
paid. In particular, since $\mathcal{T}_{2}$ is a maximum matching,
all components in $\bar{F}_{2}^{\star}$ are sufficiently paid. In
the next step
we ensure that all components of~$\bar{F}_{3}^{\star}$
are sufficiently paid. By construction, every element of $\bar{F}_{3}^{\star}$,
that is not sufficiently paid yet, intersects at least one tile of
$\mathcal{T}_{3}$. We can thus consider the components of $G(\mathcal{T}_{3})$
one by one and make sure to sufficiently pay each element of $\bar{F}_{3}^{\star}$
that intersects the considered component of~$G(\mathcal{T}_{3})$.

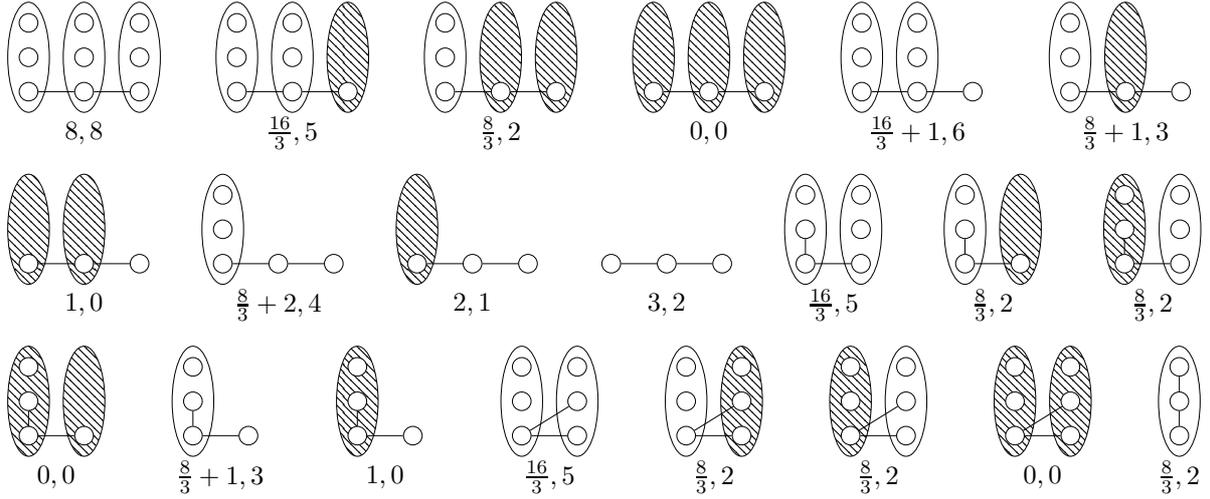
\begin{figure}[t]
\begin{centering}
\resizebox{\textwidth}{!}{%
\begin{tikzpicture}
 
   \def \dxa {3.003};
 
 \begin{scope}[shift={(0,0)}]
  \node (n00_0) at (0,0) [node] {};    
  \node (n00_1) at (0,.5) [node] {};      
  \node (n00_2) at (0,1) [node] {};   
  \draw (0,.5) ellipse (.3 and .8);
  \node (n00_3) at (.8,0) [node] {};    
  \node (n00_4) at (.8,.5) [node] {};      
  \node (n00_5) at (.8,1) [node] {};   
  \draw (.8,.5) ellipse (.3 and .8);
  \node (n00_6) at (1.6,0) [node] {};    
  \node (n00_7) at (1.6,.5) [node] {};      
  \node (n00_8) at (1.6,1) [node] {};   
  \draw (1.6,.5) ellipse (.3 and .8);
  \draw [arc] (n00_0) to (n00_3);
  \draw [arc] (n00_3) to (n00_6);
  \node (label00) at (.8,-.6) {$8,8$};
 \end{scope}

 \begin{scope}[shift={(\dxa,0)}]
  \node (n01_0) at (0,0) [node] {};    
  \node (n01_1) at (0,.5) [node] {};      
  \node (n01_2) at (0,1) [node] {};   
  \draw (0,.5) ellipse (.3 and .8);
  \node (n01_3) at (.8,0) [node] {};    
  \node (n01_4) at (.8,.5) [node] {};      
  \node (n01_5) at (.8,1) [node] {};   
  \draw (.8,.5) ellipse (.3 and .8);   
  \draw [pattern=north west lines] (1.6,.5) ellipse (.3 and .8);
  \node (n01_6) at (1.6,0) [node] {}; 
  \draw [arc] (n01_0) to (n01_3);
  \draw [arc] (n01_3) to (n01_6);
  \node (label01) at (.8,-.6) {$\frac{16}{3},5$};
 \end{scope}
 
 \begin{scope}[shift={(2*\dxa,0)}]
  \node (n03_0) at (0,0) [node] {};    
  \node (n03_1) at (0,.5) [node] {};      
  \node (n03_2) at (0,1) [node] {};   
  \draw (0,.5) ellipse (.3 and .8);
  \draw [pattern=north west lines] (.8,.5) ellipse (.3 and .8);
  \node (n03_3) at (.8,0) [node] {};    
  \draw [pattern=north west lines] (1.6,.5) ellipse (.3 and .8);
  \node (n03_6) at (1.6,0) [node] {}; 
  \draw [arc] (n03_0) to (n03_3);
  \draw [arc] (n03_3) to (n03_6);
  \node (label03) at (.8,-.6) {$\frac{8}{3},2$};
 \end{scope}
 
 \begin{scope}[shift={(3*\dxa,0)}]
  \draw [pattern=north west lines] (0,.5) ellipse (.3 and .8);
  \node (n06_0) at (0,0) [node] {};    
  \draw [pattern=north west lines] (.8,.5) ellipse (.3 and .8);
  \node (n06_3) at (.8,0) [node] {};
  \draw [pattern=north west lines] (1.6,.5) ellipse (.3 and .8);    
  \node (n06_6) at (1.6,0) [node] {}; 
  \draw [arc] (n06_0) to (n06_3);
  \draw [arc] (n06_3) to (n06_6);
  \node (label06) at (.8,-.6) {$0,0$};
 \end{scope}

 \begin{scope}[shift={(4*\dxa,0)}]
  \node (n02_0) at (0,0) [node] {};    
  \node (n02_1) at (0,.5) [node] {};      
  \node (n02_2) at (0,1) [node] {};   
  \draw (0,.5) ellipse (.3 and .8);
  \node (n02_3) at (.8,0) [node] {};    
  \node (n02_4) at (.8,.5) [node] {};      
  \node (n02_5) at (.8,1) [node] {};   
  \draw (.8,.5) ellipse (.3 and .8);   
  \node (n02_6) at (1.6,0) [node] {}; 
  \draw [arc] (n02_0) to (n02_3);
  \draw [arc] (n02_3) to (n02_6);
  \node (label02) at (.8,-.6) {$\frac{16}{3}+1,6$};
 \end{scope}

 \begin{scope}[shift={(5*\dxa,0)}]
  \node (n04_0) at (0,0) [node] {};    
  \node (n04_1) at (0,.5) [node] {};      
  \node (n04_2) at (0,1) [node] {};   
  \draw (0,.5) ellipse (.3 and .8);
  \draw [pattern=north west lines] (.8,.5) ellipse (.3 and .8);
  \node (n04_3) at (.8,0) [node] {};    
  \node (n04_6) at (1.6,0) [node] {}; 
  \draw [arc] (n04_0) to (n04_3);
  \draw [arc] (n04_3) to (n04_6);
  \node (label04) at (.8,-.6) {$\frac{8}{3}+1,3$};
 \end{scope}
 
 \def \dxb {2.8};

 \begin{scope}[shift={(0,-2.5)}]
  \draw [pattern=north west lines] (0,.5) ellipse (.3 and .8);
  \node (n07_0) at (0,0) [node] {};    
  \draw [pattern=north west lines] (.8,.5) ellipse (.3 and .8);
  \node (n07_3) at (.8,0) [node] {};
  \node (n07_6) at (1.6,0) [node] {}; 
  \draw [arc] (n07_0) to (n07_3);
  \draw [arc] (n07_3) to (n07_6);
  \node (label07) at (.8,-.6) {$1,0$};
 \end{scope}
 
 \begin{scope}[shift={(\dxb,-2.5)}]
  \node (n05_0) at (0,0) [node] {};    
  \node (n05_1) at (0,.5) [node] {};      
  \node (n05_2) at (0,1) [node] {};   
  \draw (0,.5) ellipse (.3 and .8);
  \node (n05_3) at (.8,0) [node] {};    
  \node (n05_6) at (1.6,0) [node] {}; 
  \draw [arc] (n05_0) to (n05_3);
  \draw [arc] (n05_3) to (n05_6);
  \node (label05) at (.8,-.6) {$\frac{8}{3}+2,4$};
 \end{scope}

 \begin{scope}[shift={(2*\dxb,-2.5)}]
  \draw [pattern=north west lines] (0,.5) ellipse (.3 and .8);
  \node (n08_0) at (0,0) [node] {};    
  \node (n08_3) at (.8,0) [node] {};
  \node (n08_6) at (1.6,0) [node] {}; 
  \draw [arc] (n08_0) to (n08_3);
  \draw [arc] (n08_3) to (n08_6);
  \node (label08) at (.8,-.6) {$2,1$};
 \end{scope}

 \begin{scope}[shift={(3*\dxb,-2.5)}]
  \node (n09_0) at (0,0) [node] {};    
  \node (n09_3) at (.8,0) [node] {};
  \node (n09_6) at (1.6,0) [node] {}; 
  \draw [arc] (n09_0) to (n09_3);
  \draw [arc] (n09_3) to (n09_6);
  \node (label09) at (.8,-.6) {$3,2$};
 \end{scope}

 \begin{scope}[shift={(4*\dxb,-2.5)}]
  \node (n10_0) at (0,0) [node] {};    
  \node (n10_1) at (0,.5) [node] {};      
  \node (n10_2) at (0,1) [node] {};   
  \draw (0,.5) ellipse (.3 and .8);
  \node (n10_3) at (.8,0) [node] {};    
  \node (n10_4) at (.8,.5) [node] {};      
  \node (n10_5) at (.8,1) [node] {};   
  \draw (.8,.5) ellipse (.3 and .8);
  \draw [arc] (n10_0) to (n10_3);
  \draw [arc] (n10_1) to (n10_0);
  \node (label10) at (.4,-.6) {$\frac{16}{3},5$};
 \end{scope}

 \def \dxbb {2.3};

 \begin{scope}[shift={(4*\dxb + \dxbb,-2.5)}]
  \node (n11_0) at (0,0) [node] {};    
  \node (n11_1) at (0,.5) [node] {};      
  \node (n11_2) at (0,1) [node] {};   
  \draw (0,.5) ellipse (.3 and .8);
  \draw [pattern=north west lines] (.8,.5) ellipse (.3 and .8);
  \node (n11_3) at (.8,0) [node] {};    
  \draw [arc] (n11_0) to (n11_3);
  \draw [arc] (n11_1) to (n11_0);
  \node (label11) at (.4,-.6) {$\frac{8}{3},2$};
 \end{scope}

 \begin{scope}[shift={(4*\dxb + 2*\dxbb,-2.5)}]
  \draw [pattern=north west lines] (0,.5) ellipse (.3 and .8);
  \node (n13_0) at (0,0) [node] {};    
  \node (n13_1) at (0,.5) [node] {};      
  \node (n13_2) at (0,1) [node] {};   
  \node (n13_3) at (.8,0) [node] {};    
  \node (n13_4) at (.8,.5) [node] {};      
  \node (n13_5) at (.8,1) [node] {};   
  \draw (.8,.5) ellipse (.3 and .8);
  \draw [arc] (n13_0) to (n13_3);
  \draw [arc] (n13_1) to (n13_0);
  \node (label13) at (.4,-.6) {$\frac{8}{3},2$};
 \end{scope}
 
 \def \dxc {2.37};

 \begin{scope}[shift={(0,-5)}]
  \draw [pattern=north west lines] (0,.5) ellipse (.3 and .8);
  \node (n14_0) at (0,0) [node] {};    
  \node (n14_1) at (0,.5) [node] {};      
  \node (n14_2) at (0,1) [node] {};   
  \draw [pattern=north west lines] (.8,.5) ellipse (.3 and .8);
  \node (n14_3) at (.8,0) [node] {};    
  \draw [arc] (n14_0) to (n14_3);
  \draw [arc] (n14_1) to (n14_0);
  \node (label14) at (.4,-.6) {$0,0$};
 \end{scope}
 
 \begin{scope}[shift={(\dxc,-5)}]
  \node (n12_0) at (0,0) [node] {};    
  \node (n12_1) at (0,.5) [node] {};      
  \node (n12_2) at (0,1) [node] {};   
  \draw (0,.5) ellipse (.3 and .8);
  \node (n12_3) at (.8,0) [node] {};    
  \draw [arc] (n12_0) to (n12_3);
  \draw [arc] (n12_1) to (n12_0);
  \node (label12) at (.4,-.6) {$\frac{8}{3}+1,3$};
 \end{scope}

 \begin{scope}[shift={(2*\dxc,-5)}]
  \draw [pattern=north west lines] (0,.5) ellipse (.3 and .8);
  \node (n15_0) at (0,0) [node] {};    
  \node (n15_1) at (0,.5) [node] {};      
  \node (n15_2) at (0,1) [node] {};   
  \node (n15_3) at (.8,0) [node] {};    
  \draw [arc] (n15_0) to (n15_3);
  \draw [arc] (n15_1) to (n15_0);
  \node (label15) at (.4,-.6) {$1,0$};
 \end{scope}

 \begin{scope}[shift={(3*\dxc,-5)}]
  \node (n16_0) at (0,0) [node] {};    
  \node (n16_1) at (0,.5) [node] {};      
  \node (n16_2) at (0,1) [node] {};   
  \draw (0,.5) ellipse (.3 and .8);
  \node (n16_3) at (.8,0) [node] {};    
  \node (n16_4) at (.8,.5) [node] {};      
  \node (n16_5) at (.8,1) [node] {};   
  \draw (.8,.5) ellipse (.3 and .8);
  \draw [arc] (n16_0) to (n16_3);
  \draw [arc] (n16_0) to (n16_4);
  \node (label16) at (.4,-.6) {$\frac{16}{3},5$};
 \end{scope}

 \begin{scope}[shift={(4*\dxc,-5)}]
  \node (n17_0) at (0,0) [node] {};    
  \node (n17_1) at (0,.5) [node] {};      
  \node (n17_2) at (0,1) [node] {};   
  \draw (0,.5) ellipse (.3 and .8);        
  \draw [pattern=north west lines] (.8,.5) ellipse (.3 and .8);
  \node (n17_3) at (.8,0) [node] {};    
  \node (n17_4) at (.8,.5) [node] {}; 
  \node (n17_5) at (.8,1) [node] {};
  \draw [arc] (n17_0) to (n17_3);
  \draw [arc] (n17_0) to (n17_4);
  \node (label17) at (.4,-.6) {$\frac{8}{3},2$};
 \end{scope}

 \begin{scope}[shift={(6*\dxc,-5)}]
  \draw [pattern=north west lines] (0,.5) ellipse (.3 and .8);
  \node (n18_0) at (0,0) [node] {};    
  \node (n18_1) at (0,.5) [node] {};      
  \node (n18_2) at (0,1) [node] {};           
  \draw [pattern=north west lines] (.8,.5) ellipse (.3 and .8);
  \node (n18_3) at (.8,0) [node] {};    
  \node (n18_4) at (.8,.5) [node] {}; 
  \node (n18_5) at (.8,1) [node] {};
  \draw [arc] (n18_0) to (n18_3);
  \draw [arc] (n18_0) to (n18_4);
  \node (label18) at (.4,-.6) {$0,0$};
 \end{scope}
 
 \begin{scope}[shift={(5*\dxc,-5)}]
  \draw [pattern=north west lines] (0,.5) ellipse (.3 and .8);
  \node (n19_0) at (0,0) [node] {};    
  \node (n19_1) at (0,.5) [node] {};      
  \node (n19_2) at (0,1) [node] {};   
  \node (n19_3) at (.8,0) [node] {};    
  \node (n19_4) at (.8,.5) [node] {};      
  \node (n19_5) at (.8,1) [node] {};   
  \draw (.8,.5) ellipse (.3 and .8);
  \draw [arc] (n19_0) to (n19_3);
  \draw [arc] (n19_0) to (n19_4);
  \node (label19) at (.4,-.6) {$\frac{8}{3},2$};
 \end{scope}

 \begin{scope}[shift={(7*\dxc,-5)}]
  \draw (0,.5) ellipse (.3 and .8);
  \node (n20_0) at (0,0) [node] {};    
  \node (n20_1) at (0,.5) [node] {};      
  \node (n20_2) at (0,1) [node] {};           
  \draw [arc] (n20_0) to (n20_1);
  \draw [arc] (n20_1) to (n20_2);
  \node (label20) at (0,-.6) {$\frac{8}{3},2$};
 \end{scope}

\end{tikzpicture}
}
\par\end{centering}

\caption{Possible intersections of components of $G(\mathcal{T}_{3})$
(arcs) and $G^{\star}$ (ellipses). 
Shaded components have been sufficiently paid previously. Configurations are labeled by the available gold $C^{+}$
and the required gold $C^{-}$. Symmetrical configurations
are omitted.\label{fig:intersections}}
\end{figure}

Consider a component of $G(\mathcal{T}_{3})$ involving the three
symbols $f_{1},f_{2},f_{3}$ (cf.~Figure~\ref{fig:intersections}
 in the following). Let $\mathcal{C}_{3}\subseteq\bar{F}_{3}^{\star}$
be the set of components of size~3 in $G^{\star}$ that involve at least one
of these symbols and have not yet been sufficiently paid (i.e., still
hold their gold). Further, let $N_n$ be the number of symbols among~$\{f_1, f_2, f_3\} \cap F_{\geq 4}^{\star}$ that are not yet sufficiently paid. Since all components in $\bar{F}_{2}^{\star}$
are sufficiently paid, the gold we have available is at least $C^{+}\geq\frac{4}{3}(2|\mathcal{C}_{3}|+\frac{3}{4}N_n)$. We ensure that (at least) two symbols
among $f_{1},f_{2},f_{3}$ are paid, as well as all other symbols appearing
in~$\mathcal{C}_{3}$. In this way, each component in $\mathcal{C}_{3}$
is sufficiently paid. Note that this preserves our invariant that each element of $\bar{F}_{2}^{\star}\cup\bar{F}_{3}^{\star} \cup F_{\geq 4}^{\star}$ is either sufficiently paid, or still holds its gold.
The cost for paying the symbols $f_{1},f_{2},f_{3}$
is at most 2. 
Since in addition to~$f_1,f_2,f_3$ there are~$3|\mathcal{C}_3| + N_n - 3$ symbols needing pay in $\bigcup_{C \in \mathcal{C}_3}C \cup \{f_1, f_2, f_3\}$, and because $|\mathcal{C}_{3}|\leq3$, the total cost is 
\[
C^{-} \leq 3|\mathcal{C}_{3}|+N_n- 1\leq \frac{8}{3}|\mathcal{C}_3|+N_n \leq C^{+}.
\]

At this point, we have sufficiently paid all components in $\bar{F}_{2}^{\star}\cup\bar{F}_{3}^{\star}$ using gold only from these components.
This means that all remaining symbols that are not sufficiently paid yet have at least $\frac{4}{3}\cdot\frac{4-1}{4}=1$
gold available, which we can use to pay these symbols themselves. Now all elements of $\bar{F}_{2}^{\star}\cup\bar{F}_{3}^{\star} \cup F_{\geq 4}^{\star}$ have been sufficiently paid and the proof is
complete.\end{proof}

Our analysis of \autoref{alg:approximation} is tight in three different spots: (i) A path of length $1$ in the graph $H$ defined above that visits a component of size 2 and a component of size 3 of the optimum solution $\mathcal{T}$ may lead to 4 tiles in our solution compared to the 3 tiles required in the optimum solution, i.e., Equations~(\ref{eq:cminus}) and~(\ref{eq:cplus}) coincide if $N_2^{\mathrm{end}}=N_3^{\mathrm{end}}=1$ and all other terms vanish.
(ii) The first intersection of a component of $G(\mathcal{T}_3)$ with components of $G^\star$ illustrated in Figure~\ref{fig:intersections} may lead to 8 tiles in our solution compared to the 6 tiles required in the optimum solution.
(iii) Each symbol of a component of size 4 in $G^\star$ might result in a single tile for this symbol only, in which case the optimum solution requires 3 tiles for the symbols of the component, while our solution requires 4 tiles.
To improve \autoref{alg:approximation} we have to address each of these three bottlenecks.
For (i), we either would have to alter the matching $\mathcal{T}_2$ to prevent the described situation, or combine the analysis to account for the loss in other places.
The aspect (ii) can easily be prevented by employing a more sophisticated set packing algorithm (e.g., the ($4/3+\varepsilon)$-approximation of Cygan~\cite{Cygan13}).
Finally, to avoid (iii), we would need to pack sets of size 4 similarly to our packing of sets of size 3. In addition to requiring one more level of analysis, this would also complicate the other levels, as we would have to include sets of size 4 in our reasoning there. 

\section{Bounded number of scenarios}\label{section:numberofscenarios}
In this section, we prove that \mfts can be solved in polynomial time when the number~$|\S|$ of scenarios is some constant. For convenience, for the course of this section, we switch to the decision variant of \mfts. That is, we equip each instance~$(F, \S)$ of \mfts\ with an additional integer~$\ell$, and we ask, whether there is a feasible tileset for $\S$ with at most~$\ell$ tiles. Clearly, solving the decision variant in polynomial time implies that also the optimization variant is solvable in polynomial time. We provide an algorithm that solves any instance~$(F,\S,\ell)$ in time $f(|\S|)|(F,\S,\ell)|^c$, i.e., in time $\Oh(|(F,\S,k)|^c)$ for bounded values of~$|\S|$. In other words, \mfts is fixed-parameter tractable with respect to the number of scenarios. 

Our algorithm works by first translating the input instance~$(F,\S,\ell)$ into an integer linear program (ILP) in such a way that the ILP is feasible (i.e., contains at least one integer point) if and only if~$(F,\S,\ell)$ admits a feasible tileset with at most~$\ell$~tiles. The ILP uses~$\Oh(|\S|^{|\S|})$ variables. Lenstra~\cite{Lenstra83} proved that deciding feasibility of any ILP is fixed-parameter tractable with respect to the number of variables; the currently fastest algorithm was obtained by Frank and Tardos~\cite{FT87}, modifying an algorithm by Kannan~\cite{Kannan87}.
\begin{theorem}[Frank and Tardos~\cite{FT87}]\label{theorem:ilpfeasibility:kannan}
  In $\Oh^*(p^{\Oh(p)})$ time we can decide whether a given ILP with $p$~variables is feasible.
\end{theorem}
Using this, we can prove the following result.

\begin{theorem}\label{theorem:boundedscenarios:fpt}
\mfts on instances with at most~$k$ scenarios can be solved in time~$\Oh^*(k^{\Oh(k^{k+1})})$.
\end{theorem}
Intuitively, a bounded number of scenarios also implies a bound on the number of different subsets of scenarios in which a tile can appear. Thus, one would like to forget the actual identities of the symbols and only remember \emph{how many} symbols appear, say, exactly in scenarios~$S_1$,~$S_5$, and~$S_6$. It appears, however, that grouping symbols in this way is insufficient since symbols from the same group can nevertheless have different patterns for how they are provided by tiles: E.g., one tile could provide such a symbol in all three scenarios~$S_1$,~$S_5$, and~$S_6$, whereas other symbols of the same group might need three separate tiles for~$S_1$,~$S_5$, and~$S_6$. To cope with this, the constructed ILP has separate variables for all partitions of scenario subsets as well as variables for all ways of using a tile (recall that a tile has two symbols, meaning that it has two disjoint subsets of the scenario that express when either symbol is provided by the tile).

\begin{proof}[Proof of \autoref{theorem:boundedscenarios:fpt}]
We formulate \mfts as an ILP and employ Kannan's algorithm. 
Intuitively, each tile contributes both of its symbols to different (disjoint) subsets of the scenarios. 
For example, if we have $5$ scenarios, a tile might contribute one of its symbols to scenarios~$1$ and $4$, the other to scenarios~$3$ and~$5$, and neither to scenario~$2$.
Each tile is associated with such a pattern of how it contributes to scenarios, and one part of the variables of our ILP track the number of tiles having each of the possible patterns.
On the other hand, each symbol has a pattern associated with it, depending on which occurrences of the symbol are provided by the same tile.
In our example, a symbol appearing in scenarios $1$, $2$, and $4$ might be provided by the same tile in scenarios $1$ and $4$, and by a different tile in scenario~$2$.
The remaining variables of the ILP track the number of symbols having each of the possible patterns.
We provide exchange arguments to show that enforcing correct totals for these variables  by linear constraints is sufficient to ensure that a feasible assignment of tiles to symbols exists for each scenario.

\emph{ILP formulation.} To make our description precise, let an instance~$(F,\S,\ell)$ with $k$ scenarios $\S=\{S_1,\ldots,S_k\}$ be given. For brevity, we refer to a subset of~$\S$ by the corresponding index set. For every subset~$I\subseteq [k]$ of scenarios we count the number of symbols that occur exactly in these scenarios and denote this number by~$c_I = |\bigcap_{i\in I}S_i \setminus \bigcup_{i\notin I} S_i|$. The family of all partitions of $I$ is denoted by $\partition(I)$. The ILP is constructed as follows. 

\begin{enumerate}
  
 \item For each set~$I\subseteq[k]$ and each partition~$\I = \{I_1,\ldots,I_s\} \in \partition(I)$ we introduce a variable~$y_\I$. 
 The intention is that variable~$y_\I$ counts the number of symbols that occur (exactly) in scenarios~$I:=I_1\cup\ldots\cup I_s$ and have pattern \I associated with them, in the following way: 
 Exactly~$s$ tiles, say,~$T_1,\ldots,T_s$, are used for such a symbol and the symbol is provided by tile~$T_i$ in the scenarios ~$I_i$.
 
 For each~$I$ we add a constraint that enforces the total number of patterns to equal the  number~$c_I$ of symbols that occur in the scenarios $I$:
 \begin{align*}
  c_I=\sum_{\I\in\partition(I)} y_\I\qquad \forall I\subseteq [k].
 \end{align*}
 
 For example, if~$I = \{1,2,3\}$, the following variables are created:
 \[
 y_{\{\{1,2,3\}\}}, y_{\{\{1,2\},\{3\}\}}, y_{\{\{1,3\},\{2\}\}}, y_{\{\{2,3\},\{1\}\}}, y_{\{\{1\},\{2\},\{3\}\}}.
 \]
 
 The number of~$y$-variables equals the number of subpartitions of the set~$[k]$. This is upper bounded by~$k^k+1$: We can~$k$-color all subpartitions other than the partition into singletons by using color~$k$ for all unused elements and colors~$1,\ldots,k-1$ for the elements of each set in the partition (only the partition into singletons has~$k$ sets). Thus, we get an injective mapping of all but one subpartition into the~$k$ colorings of~$[k]$; this gives a total of~$k^k+1$.
 \item For the tiles, we introduce variables~$x_{I,J}$ for all~$I,J\subseteq [k]$ with~$I\cap J=\emptyset$ and~$I\cup J\neq\emptyset$; for convenience we identify~$x_{I,J}=x_{J,I}$. 
Intuitively, the variable~$x_{I,J}$ stands for the number of tiles that provide one of their symbols for scenarios~$I$ and the other symbol for scenarios~$J$.
 
For example, for~$k=3$ we create the following variables:
 \begin{align*}
&x_{\emptyset,\{1\}}, x_{\emptyset,\{2\}}, x_{\emptyset,\{3\}}, x_{\emptyset,\{1,2\}}, x_{\emptyset,\{1,3\}}, x_{\emptyset,\{2,3\}},x_{\emptyset,\{1,2,3\}},\\
&x_{\{1\},\{2\}},x_{\{1\},\{3\}},x_{\{1\},\{2,3\}},x_{\{2\},\{3\}},x_{\{2\},\{1,3\}},x_{\{3\},\{1,2\}}\text{.}
 \end{align*}
 
 The number of~$x$-variables is~$\frac{3^k-1}{2}$ corresponding to all partitions of~$[k]$ into three sets (i.e.,~$I$,~$J$, and $[k]\setminus (I \cup J)$), without~$I=J=\emptyset$, and identifying~$x_{I,J}$ with~$x_{J,I}$.

 We add constraints that enforce that the number of tiles of each pattern match the sum of the corresponding $y$-variables. Concretely, we add
 \begin{align}
  \sum_{\substack{I \subseteq J\subseteq[k]\\ \J\in\partition(J)\\ I\in\J}}y_{\J}=\sum_{J\subseteq[k]\setminus I}x_{I,J} \qquad \forall  I\subseteq [k], I \neq \emptyset .\label{constraint:xy}
 \end{align}
 
 We compare the number of tiles that provide one of their symbols for scenarios in $I$ with the number of symbols that have~$I$ in their pattern. 
 For the set of scenarios~$J$ such symbols appear in we must have~$I\subseteq J\subseteq[k]$, and we need partitions~$\J\in\partition(J)$ that contain~$I$.
 
 \item As a final constraint we enforce that the total number of used tiles is no more than~$\ell$. To this end, we simply sum over all~$x$-variables and add
 \begin{align*}
  \frac{1}{2} \sum_{\substack{I,J\subseteq [k]\\I\cap J=\emptyset\\I \cup J \neq \emptyset}}x_{I,J}\leq\ell.
 \end{align*}
\end{enumerate}
This completes our construction. We use~$p\leq k^k+1+\frac{3^k-1}{2}=\Oh(k^k)$ variables and, thus, Kannan's algorithm decides feasibility of our ILP in time~$\Oh^*(p^{\Oh(p)})=\Oh^*((k^k)^{\Oh(k^k)})=\Oh^*(k^{\Oh(k^{k+1})})$.

\emph{Correctness.} 
First assume that the given instance~$(F,\S,\ell)$ of \mfts admits a feasible tileset \T of minimum cardinality $|\T|\leq\ell$. 
Since~$\T$ is feasible for each scenario~$S_i\in\S$, we may let~$\varphi_i\colon S_i\to\T$ be an injective function that assigns each symbol in~$S_i$ a unique tile in~$\T$ that can provide it. We specify feasible values for the~$x$- and~$y$-variables.

\begin{enumerate}
 \item \emph{$x$-variables.} 
 Each tile~$T\in\T$ has two symbols, say,~$T=\{s,s'\}$, and, hence, for each~$i\in[k]$ it is the image of at most one of~$s$ and~$s'$. Formally, let
    \begin{align*}
    I&:=\{i\in[k]\mid \varphi_i(s)=T\},\\
    J&:=\{j\in[k]\mid \varphi_j(s')=T\}.
    \end{align*}
That is, the set~$I$ contains all scenarios for which tile~$T$ provides symbol~$s$, and~$J$ is the analogue for symbol~$s'$. Since the functions~$\varphi_i$ are injective, we must have that~$I\cap J=\emptyset$.
  
We have~$I\cup J\neq\emptyset$ as otherwise~$T$ would not be used for any scenario, contradicting the minimality of~$\T$. We say that tile~$T$ has \emph{pattern}~$\{I,J\}$.

For each~$I,J\subseteq[k]$ with~$I\cap J=\emptyset$ and~$I\cup J\neq\emptyset$, we set~$x_{I,J}$ to the number of tiles with pattern~$\{I,J\}$. Clearly, the constraint forcing the total value of the~$x$-variables to be at most~$\ell$ is fulfilled since~$|\T|\leq \ell$.

 \item \emph{$y$-variables.} Similarly to the tiles in~$T$ we determine a pattern for each symbol~$s\in F$. We let~$\T(s):=\{T\in \T\mid \exists i\in[k]:\varphi_i(s)=T\} = \{T_1,\dots,T_r\}$, i.e., the set of tiles that provide~$s$ in at least one scenario. Let~$I\subseteq[k]$ be the set of scenarios containing~$s$. We define a partition~$\{I_1,\ldots,I_r\}$ of~$I$ by
 \begin{align*}
  I_p:=\{i\in[k]\mid \varphi_i(s)=T_p\},
 \end{align*}
 for all $p \in [r]$.
 We say that symbol~$s$ has \emph{pattern}~$\{I_1,\ldots,I_r\}\in\partition(I)$.

 For each~$I\in[k]$ and each partition~$\I\in\partition(I)$ we set~$y_{\I}$ to the number of symbols in~$F$ with pattern~$\I$. Clearly, this fulfills the constraint that all~$y$-variables whose pattern is a partition of some set~$I\subseteq [k]$ equals the total number~$c_I$ of symbols that occur exactly among the scenarios in~$I$.
\end{enumerate}
It remains to verify that the constraint relating~$x$- and~$y$-variables is satisfied. To this end, let us fix some~$I\subseteq[k]$, $I\neq\emptyset$, and consider the constraint
 \begin{align*}
  \sum_{\substack{I\subseteq J\subseteq[k]\\\J\in\partition(J)\\I\in\J}}y_{\J}=\sum_{J\subseteq[k]\setminus I}x_{I,J}.
 \end{align*}
 
For each tile~$T\in\T$ that contributes to the right-hand-side, there must be a unique symbol~$s$ in~$F$, such that~$\varphi_i(s) = T$ if and only if~$i \in I$. For this symbol, we have~$T\in\T(s)$, the set of scenarios~$J$ containing~$s$ satisfies~$I\subseteq J\subseteq [k]$, and $I$ is part of the pattern of~$s$. 
Hence, $s$ contributes to the left-hand-side.
Conversely, if~$s$ is a symbol contributing to the left-hand-side, then~$I$ must be part of the pattern of~$s$.
This means that there is a unique tile $T\in\T$, such that~$\varphi_i(s) = T$ if and only if~$i \in I$.
This tile has $I$ in its pattern and thus contributes to the right-hand-side.
Overall, the contribution to both sides is equal, and our assignment to~$x$- and~$y$-variables is feasible, as claimed.

Now, assume that the ILP constructed from $(F,\S,\ell)$ is feasible and fix a feasible assignment to the $x$- and $y$-variables. 
We derive a feasible tileset for all scenarios in~$\S$.
The set of all symbols can be partitioned according to the scenarios~$I\subseteq[k]$ that each symbol appears in. 
The total count~$c_I$ of symbols in~$I$ is matched by the sum of~$y$-variables that are indexed by the partitions~$\I\in\partition(I)$. 
We arbitrarily assign to each symbol with scenario set~$I$ a pattern~$\I\in\partition(I)$ under the sole constraint that the total number of symbols with pattern~$\I$ matches the corresponding variable~$y_{\I}$. 
For a symbol with assigned pattern~$\I=\{I_1,\ldots,I_r\}$ the intention is to use~$r$ tiles~$T_1,\ldots,T_r$ that are each responsible for one set~$I_p\in\I$.

We will use a number of tiles that exactly matches the sum of~$x$-variables, and thereby ensure that the final tileset has cardinality at most~$\ell$. We do not pick symbols for each tile but, according to the~$x$-variables, we pick for each tile two disjoint sets of scenarios in which its two symbols will be used. Concretely, exactly~$x_{I,J}$ tiles will be used in~$I$-scenarios for one symbol and in~$J$ scenarios for their other symbol, i.e., we use~$x_{I,J}$ tiles of pattern~$\{I,J\}$. Recall that~$I\cap J=\emptyset$ and that the sum of these variables does not exceed the maximum number of allowed tiles $\ell$.

Finally, we assign symbols to tiles according to symbol and tile patterns in a canonical way. 
Specifically, symbols whose pattern contains some fixed~$I\subseteq[k]$ are assigned to tiles that contain~$I$ in their pattern. 
By constraint~\cref{constraint:xy} the number of symbols and the number of tiles are equal. 
Note that each tile is used for two disjoint sets~$I,J\subseteq[k]$ and each variable~$x_{I,J}$ appears in two~\cref{constraint:xy}-constraints (for~$I$ and for~$J$). 
Thus, each tile with pattern~$\{I,J\}$ is assigned two symbols, one requiring the tile for the scenarios in~$I$ and the other requiring it the ones in~$J$. 
Similarly, a symbol with pattern~$\I=\{I_1,\ldots,I_r\}$ contributes to~$r$ constraints~\cref{constraint:xy}, one for each~$I_1,\ldots,I_r$. 
Accordingly, these constraints enforce the correct sum of the corresponding variables~$x_{I_1,\cdot},\ldots,x_{I_r,\cdot}$. (Recall that we identified~$x_{I,J}$ with~$x_{J,I}$.)

We argue that the constructed tileset is indeed feasible for all scenarios~$S_i\in\S$.
Consider any symbol~$s\in S_i$ with pattern $\J$.
Since~$s$ appears in $S_i$, we have~$i\in I \in \J$ for some set~$I$.
By the above, we know that there is a tile~$T$ containing~$s$ that has~$I$ as a part of its pattern~$\{I,J\}$.
Since, by definition,~$I\cap J = \emptyset$, we have~$i\notin J$ and may safely use~$T$ for symbol~$s$ in scenario~$S_i$.
\end{proof}

\section{Bounded number of symbols}\label{section:numberofsymbols}
We now analyze the influence of the number of symbols~$|F|$ on the complexity of solving an instance $(F,\S,\ell)$ of the decision variant of \mfts. (That is, as in section \autoref{section:numberofscenarios}, we want to decide whether there is a feasible tileset for $\S$ with at most $\ell$ tiles.) It is easy to see that the problem becomes solvable in polynomial time when~$F$ is bounded:
The instance is trivial if~$\ell\geq|F|$ since, in that case, we can afford to dedicate a separate tile for each symbol. 
Otherwise, there are only~$\Oh(|F|^{2\ell}) \subseteq \Oh(|F|^{2|F|})$ ways to fix~$\ell$ tiles.
As mentioned in \autoref{section:hardness}, each candidate tileset can be verified by solving a bipartite matching problem for each scenario, on a graph that has an edge between each symbol in the scenario and every tile containing that symbol.
This yields an overall runtime of~$\Oh^*(|F|^{2|F|})$, and, hence, fixed-parameter tractability in~$|F|$.
Using structural insights of \autoref{section:structure} we are able to improve on this naive running time.

\begin{theorem}\label{theorem:boundedsymbols:dp}
Instances~$(F,\S,\ell)$ of the decision variant of \mfts can be solved in time~$\Oh^*(3^{|F|})$.
\end{theorem}

Note that, as every symbol occurs in a scenario, $\ell \geq |F|/2$. Hence, \autoref{theorem:boundedsymbols:dp} gives a fixed-parameter algorithm also for parameter~$\ell$.

\begin{proof}[Proof of \autoref{theorem:boundedsymbols:dp}]
We describe a dynamic programing algorithm for solving an instance $(F,\S,\ell)$.
Recall that we may assume~$\ell<|F|$; otherwise the instance is trivial.
Our algorithm uses a table~$M$ of size~$2^{|F|}$ that is indexed by subsets~$D\subseteq F$, with each entry taking integer values from~$[|F|]\cup\{-\infty\}$.
At the end of the computation, each entry~$M(D)$ will be set to~$-\infty$ if~$D \subseteq S$ for some scenario~$S\in\S$, and otherwise to the maximum integer $i \in [|F|]$ for which there is a partition of $D$ into $i$ sets $D_1,\dots,D_i$ such that no scenario contains any set in $\{D_1,\dots,D_i\}$ as a subset.

In the end, by \autoref{cor:only_components_matter}, the entry~$M(F)$ contains the maximum number of components in the graph corresponding to a  feasible tileset.
Accordingly, every corresponding tileset~$\T$ has minimum cardinality.
Hence, and since each connected component $C$ in the graph~$(F,\T)$ is composed of~$|C|-1$ tiles, the instance~$(F,\S,\ell)$ admits a tileset of size~$\ell$ if and only if~$M(F)\geq |F|-\ell$.

We fill out the entries of the table in order of increasing subset sizes.
Each entry is computed via the following recurrence. (Note that the~$1$ in the maximum taken over subsets~$D'$ of~$D$ stands for the trivial partition of~$D$ into just one set. This is the best value in case that no split into at least two sets can be found such that both sets are not subsets of scenarios.)
\[
M(D)=\begin{cases}
-\infty, & \text{if } D\subseteq S \text{ for some }S\in\S,\\
\max_{\substack{D'\subset D\\2\leq|D'|\leq |D|/2}}\{1, M(D')+M(D\setminus D')\}, & \text{otherwise.}
\end{cases}
\]
Thus, for each~$D\subseteq F$ that is not a subset of a scenario we need to compute the maximum of~$M(D')+M(D\setminus D')$ over less than~$2^{|D|}$ subsets~$D'$ of~$D$. By the well-known binomial theorem the total number of evaluations taken over all~$D\subseteq F$ can be upper bounded by~$3^{|F|}$ giving us the claimed runtime.
\end{proof}

After this fixed-parameter tractability result, and taking into account the trivial bound of~$2^{|F|}$ for the number of scenarios (giving a worst-case size of instances of~$\Oh(2^{|F|}|F|)$), it is natural to ask whether polynomial-time preprocessing can simplify input instances to size polynomial in~$|F|$. We show that this is impossible unless \containment (and the polynomial hierarchy collapses). More generally, we prove that for the restricted case~$d$-\mfts, where scenarios have size at most~$d$, no polynomial-time algorithm can achieve a size of~$\Oh(k^{d-\varepsilon})$. Note that this restricted case has an essentially matching upper bound of~$|\S|< (|F|+1)^d=\Oh(|F|^d)$.\footnote{A compression to $\Oh(|F|^d)$~size can be achieved by specifying one bit for each possible scenario in~$\S$ and setting it to one if the scenario is present and zero otherwise.} As a consequence there is no reduction to size polynomial in~$|F|$ for the general \mfts problem: Any size~$\Oh(k^c)$ preprocessing for \mfts could be used for~$d$-\mfts, for any~$d>c$, and violate the lower bound.

\begin{theorem}\label{thm:encoding}
Let~$d\geq 3$ and $\varepsilon$ be a positive real. 
There is no polynomial-time algorithm that reduces every instance of~$d$-\mfts
to an equivalent instance (possibly of a different problem) of size~$\Oh(|F|^{d-\varepsilon})$, 
unless \containment.
\end{theorem}

To prove~\autoref{thm:encoding} we employ a similar result by Dell and Marx~\cite{DellM12} for \pdsm, which is defined as follows.\footnote{Dell and Marx called this problem \textsc{Perfect $d$-Set Matching}.}

\introduceproblem{\pdsm}{A universe~$X$ and a family $\mathcal{C}$ of $d$-element sets $C\in\binom{X}{d}$.}
{Is there an \emph{exact $d$-set cover} for $X$, i.e., a partition of $X$ into a family $\mathcal{C}'\subseteq\mathcal{C}$ of disjoint sets?}

Note that the original result by Dell and Marx~\cite{DellM12} is given in terms of the size~$k$ of an exact $d$-set cover.
Clearly, $k = \frac{|U|}{d}$ and, thus, we have~$\Oh(k^{d-\varepsilon})=\Oh(|U|^{d-\varepsilon})$ and may instead phrase the result in terms of~$|U|$.
Furthermore, their result builds on work by Dell and van Melkebeek~\cite{DellM14} and, thus, extends to any polynomial time algorithms (rather than just problem kernels as mentioned there) whose output instances can be with respect to a different problem.
We give the following paraphrased version of the result.

\begin{theorem}[Dell and Marx~\cite{DellM12}]\label{theorem:dellmarx}
Let~$d\geq3$ and~$\varepsilon$ be a positive real.  
There is no polynomial-time algorithm that reduces every instance~$(U,\h)$ of \pdsm to an equivalent instance of size $\Oh(|U|^{d-\varepsilon})$ (possibly with respect to a different problem), 
unless \containment.
\end{theorem}

The following lemma, together with \autoref{theorem:dellmarx}, directly implies \autoref{thm:encoding}.

\begin{lemma}\label{lem:red-xts-gen}
There is a polynomial-time reduction from \pdsm to \mfts such that instances~$(X,\mathcal{C})$ are mapped to instances~$(F,\S,\ell)$ with~$F=X$ and scenario size at most~$d$.
\end{lemma}

\begin{proof}[Proof sketch of \autoref{lem:red-xts-gen}]
  The proof is similar to the proof of \NP-hardness in \autoref{thm:restricted-hard}. Given an instance of \pdsm{} with universe~$X$ and a family~$\mathcal{C}$ we construct an instance~$(F, \S, (d - 1)n/d)$ of \mfts\ with~$F = X$ and $n = |X|$. Applying the equivalence of finding a feasible tileset of size $(d-1)n/d$ and finding an admissible partition for $\S$ of size $n/d$ then gives \autoref{lem:red-xts-gen}.
  
  To construct the instance of \mfts, we simply set~$\S = \binom{X}{d - 1} \cup (\binom{X}{d} \setminus \mathcal{C})$. Similarly to the reduction used for \autoref{thm:restricted-hard}, the scenarios~$\binom{X}{d - 1}$ enforce that every admissible partition contains only parts of size exactly~$d$. The constraints~$\binom{X}{d} \setminus \mathcal{C}$ enforce that only sets of~$\mathcal{C}$ occur in an admissible partition. Hence, each admissible partition is also an exact $d$-set cover and vice-versa.
\end{proof}

We now consider a more general setting: In the \gmfts problem we are also given a set of symbols and a set of scenarios, but here each scenario may be a \emph{multi-set} of symbols (or, equivalently, each scenario is a function~$S\colon F\to\N$ indicating the number of copies of each symbol~$f$ needed for~$S$). We prove that \gmfts can be solved in time~$\Oh^*(|F|^{\Oh(|F|^2)})$. Note that for this problem the solution size~$\ell$ may be much larger than~$|F|$ and similarly the number of scenarios cannot in general be bounded in~$|F|$.

\begin{theorem}\label{theorem:boundedsymbols:fpt}
\gmfts can be solved in time~$\Oh^*(|F|^{\Oh(|F|^2)})$, i.e., it is fixed-parameter tractable with respect to~$|F|$.
\end{theorem}
\begin{proof}
Let~$(F,\S,\ell)$ be an instance of \gmfts{} and let~$k:=|F|$. We will construct an integer linear program~(ILP) with~$\binom{k}{2}=\Oh(k^2)$ variables and~$\Oh^*(2^k)$ constraints that is feasible if and only if~$(F,\S,\ell)$ admits a feasible tileset with at most~$\ell$ tiles. Using Kannan's algorithm (\autoref{theorem:ilpfeasibility:kannan}) then completes the proof.

We introduce one variable~$x_{s,s'} \geq 0$ for each possible tile type, i.e., for each pair of symbols~${s,s'}\in\binom{F}{2}$.
We interpret~$x_{s,s'}$ as the number of tiles of type~${s,s'}$ that the solution will contain. 
We begin with the constraint ensuring that we do not use more than~$\ell$ tiles overall:

\begin{align*}
\sum_{\{s,s'\}\in\binom{F}{2}}x_{s,s'}\leq\ell
\end{align*}

We need to add constraints to the ILP to ensure that the resulting assignment to the~$x_{s,s'}$-variables corresponds to a feasible tileset, i.e., that each scenario~$S$ can be implemented using the corresponding numbers of tiles of each type. 
This is the case if and only if there is a matching from the symbols in~$S$ to the tiles that cover all symbols in~$S$.
Clearly, in order not to use too many variables, we do not want to compute a (one-sided perfect) matching for each scenario~$S$.
By Hall's Theorem, it is instead sufficient to ensure that for each subset~$I \subset F$ of symbols appearing at least once in scenario~$S$ there are at least that many tiles involving these symbols.
If $c_{s,S}$ denotes the number of occurrences of symbol~$s$ in scenario~$S$, we obtain the following constraints:
 \begin{align*}
  \sum_{\{s,s'\}\cap I\neq\emptyset}x_{s,s'}\geq \sum_{s \in I}c_{s,S} \qquad \forall S\in\S, \forall I\subseteq F
 \end{align*}
In total we use~$\binom{k}{2}=\Oh(k^2)$ variables and~$1+m\cdot 2^k+\binom{k}{2}$ constraints. Using Kannan's algorithm for testing feasibility of an ILP with~$p$~variables in time~$\Oh^*(p^{\Oh(p)})$ (Theorem~\ref{theorem:ilpfeasibility:kannan}) we get a total running time of~$\Oh^*(k^{\Oh(k^2)})$.
\end{proof}

\section{Conclusion}\label{section:conclusion}

We initiated the study of the \mfts problem
and exposed an interesting combinatorial structure. 
We proved the problem to be \NP-complete even in the restricted case with scenarios of size at most three and \APX-hard in general. On the positive side, we showed that the \mfts{} problem admits a 4/3-approximation algorithm and that it is fixed-parameter tractable with respect to the number of scenarios and number of symbols. 
The latter algorithm works also for the \gmfts problem where each scenario can contain multiple copies of a symbol and we believe that it can be further generalized to work also for the original assignment problem where also tiles of larger (but constant) size are allowed.
It would be interesting to see whether our other positive results transfer to this more general setting. We note that our approximation algorithm relies heavily on the structural observations from \autoref{section:structure} which do not seem to generalize well. Our integer linear program for a fixed number of scenarios does not seem easily adaptable either.

\bibliographystyle{abbrv}
\bibliography{paper}

\end{document}